\documentclass{vldb}

\usepackage{amsmath}
\usepackage{ifthen}
\usepackage{tikz} 
\usepackage{epstopdf} 
\usepackage[caption=false]{subfig}
\usepackage{figures} 
\usepackage{multirow}
\usepackage{algorithm}
\usepackage[noend]{algorithmic} 
\usepackage{pgfplots}
\usepackage{thmtools, thm-restate}
\usepackage{textcase}
\usepackage{color}
\usepackage[export]{adjustbox}
\usepackage{dsfont}
\usepackage{hyperref}
\usepackage{enumitem}
\usepackage{balance} 
\usepackage[T1]{fontenc}

\usetikzlibrary{calc, shapes,arrows,positioning}
\usetikzlibrary{shapes.misc, fit, arrows, decorations.markings}

\renewcommand{\algorithmiccomment}[1]{\bgroup\hfill//~\color{gray}#1\egroup}

\newcommand{\specialcell}[2][c]{
	\begin{tabular}[#1]{@{}c@{}}#2\end{tabular}
}

\newcommand{\norm}[1]{\left\lVert#1\right\rVert}
\newcommand{\X}[1]{X^{(#1)}}

\newcommand{\PIMB}{\pmb{\pi}_{\mathbf{e}}}

\newcommand{\mathbbm}[1]{\mathds{#1}}

\makeatletter
\def\@sect#1#2#3#4#5#6[#7]#8{%
    \ifnum #2>\c@secnumdepth
        \let\@svsec\@empty
    \else
        \refstepcounter{#1}%
        \edef\@svsec{%
            \begingroup
                        \ifnum#2>2 \noexpand#6 \fi
                \csname the#1\endcsname
            \endgroup
            \ifnum #2=1\relax .\fi
            \hskip 1em
        }%
    \fi
    \@tempskipa #5\relax
    \ifdim \@tempskipa>\z@
        \begingroup
            #6\relax
            \@hangfrom{\hskip #3\relax\@svsec}%
            \begingroup
                \interlinepenalty \@M
                \if@uchead
                    \MakeTextUppercase{#8}
                \else
                    #8%
                \fi
                \par
            \endgroup
        \endgroup
        \csname #1mark\endcsname{#7}%
        \vskip -12pt  
\addcontentsline{toc}{#1}{%
            \ifnum #2>\c@secnumdepth \else
                \protect\numberline{\csname the#1\endcsname}%
            \fi
            #7%
        }%
    \else
        \def\@svsechd{%
            #6%
            \hskip #3\relax
            \@svsec
            \if@uchead
                \uppercase{#8}%
            \else
                #8%
            \fi
            \csname #1mark\endcsname{#7}%
            \addcontentsline{toc}{#1}{%
                \ifnum #2>\c@secnumdepth \else
                    \protect\numberline{\csname the#1\endcsname}%
                \fi
                #7%
            }%
        }%
    \fi
    \@xsect{#5}\hskip 1pt
    \par
}
\makeatother

\makeatletter
\def\tagform@#1{\maketag@@@{\bfseries(\ignorespaces#1\unskip\@@italiccorr)}}
\renewcommand{\eqref}[1]{\textup{{\normalfont(\ref{#1}}\normalfont)}}
\makeatother

\definecolor{red}{HTML}{E51400} 
\definecolor{green}{HTML}{008A00} 
\definecolor{purple}{HTML}{AA00FF} 
\definecolor{darkgreen}{HTML}{004225} 

\hypersetup{
    colorlinks,  
    linkcolor=blue,
    citecolor=green,
    urlcolor=blue
}

\declaretheorem[name={Theorem}]{theorem}
\newtheorem{lemma}[theorem]{Lemma}
\renewenvironment{proof}{\textsc{Proof.}}{\qed}
\newdef{definition}{Definition}

\begin{document}


\title{A General Framework for Estimating Graphlet Statistics via Random Walk}


\numberofauthors{1} 

\author{
\alignauthor
Xiaowei Chen$^\textsuperscript{1}$, Yongkun Li$^{\textsuperscript{2}}$, Pinghui
Wang$^\textsuperscript{3}$, John C.S. Lui$^\textsuperscript{1}$\\
\affaddr{$^{1}$The Chinese University of Hong Kong}\\
        \affaddr{$^2$University of Science and Technology of China}\\
        \affaddr{$^3$Xi'an Jiaotong University}\\
        \email{$^1$\{xwchen, cslui\}@cse.cuhk.edu.hk, $^2$ykli@ustc.edu.cn,
        $^3$phwang@mail.xjtu.edu.cn}
}

\maketitle

\begin{abstract}
    Graphlets are induced subgraph patterns and  have been frequently applied to
    characterize the local topology structures of graphs across various
    domains, e.g., online social networks (OSNs) and biological networks.
    Discovering and computing graphlet statistics are highly challenging.
    First, the massive size of real-world graphs makes the exact computation of
    graphlets extremely expensive.  Secondly, the graph topology may not be
    readily available so one has to resort to web crawling using the available
    application programming interfaces (APIs).
    In this work, we propose a general and novel framework to estimate graphlet
    statistics of ``{\em any size}''. Our framework is based on collecting
    samples through consecutive steps of random walks. We derive an analytical
    bound on the sample size (via the Chernoff-Hoeffding technique) to
    guarantee the convergence of our unbiased estimator. 
    To further improve the accuracy, we introduce two novel optimization techniques
    to reduce the lower bound on the sample size.  Experimental evaluations
    demonstrate that our methods outperform the state-of-the-art method up to
    an order of magnitude both in terms of accuracy and time cost.
\end{abstract}

\section{Introduction}\label{sec:introduction}
Graphlets are defined as induced subgraph patterns in real-world
networks~\cite{Milo824}. Unlike some global properties such as
degree distribution, the frequencies of graphlets provide important statistics
to characterize the local topology structures of networks.  Decomposing networks
into small $k$-node graphlets has been a fundamental approach to characterize
the local structures of real-world complex networks.  Graphlets also have
numerous applications ranging from biology to network science.  Applications in
biology include protein detection~\cite{milenkovic2008uncovering}, biological
network comparison~\cite{przuljbio10} and disease gene
identification~\cite{Milenkovi423}.  In network science, the researchers have
applied graphlets for web spam detection~\cite{BecchettiESA}, anomaly
detection~\cite{AkogluTK14}, social network structure
analysis~\cite{UganderSFM}, and friendship recommendation,
e.g., there are two ``{\em types}'' of 3-node graphlets: (1) a 3-node line
subgraph or (2) a triangular subgraph. If 3-node line subgraphs occur with much
higher frequency than triangular subgraphs in an OSN, then we know that there
are more opportunities for us to make friendship recommendation.

\textbf{Research Problem.} In most applications, relative frequencies among
various graphlets are sufficient. 
One example is the friendship
recommendation in OSNs we have just mentioned. 
Another example is building graphlet kernels for large graph
comparison~\cite{5664}.
In this work, we focus on relative graphlet frequencies discovery and computation.
More specifically, we propose
efficient sampling methods to compute the \textit{percentage} of a specific
$k$-node graphlet type within all $k$-node graphlets in a given graph. 
The percentage of a particular $k$-node graphlet type is called the
``{\em graphlet concentration}'' or ``{\em graphlet statistics}''.

\textbf{Challenges.} 
The straightforward approach to compute the graphlet concentration is via exhaustive counting.
However, there exist a large number of graphlets even for a
moderately sized graph. For example,
Facebook~\cite{konect} in our datasets with 817K edges
has $9\times 10^9$ 4-node graphlets and $2\times 10^{12}$ 5-node graphlets.
Due to the combinatorial explosion problem, how to count graphlets efficiently is a long
standing research problem. Some techniques, such as leveraging parallelism
provided by multi-core architecture~\cite{ahmed2015icdm}, exploiting
combinatorial relationships between graphlets~\cite{hovcevar2014combinatorial},
and employing distributed systems~\cite{SuriCTC}, have been applied to speed up the
graphlet counting. However, these exhaustive counting algorithms are not scalable 
because they need to explore
each graphlet at least once. Even with those highly
tuned algorithms, exhaustive counting of graphlets has prohibitive computation cost
for real-world large graphs. An alternative approach is to adopt
``{\em sampling algorithms}'' to achieve significant speedup with acceptable
error. Several methods based on sampling have been proposed to address the
challenge of graphlet counting~\cite{jha2015path,
DBLP:journals/corr/WangTZG15, wang2014efficiently, bhuiyan2012guise,
wang2015minfer, ElenbergBTD}.  

Another challenge is the restricted access to the
complete graph data. For example, most OSNs' service
providers are unwilling to share the complete data for public use. The
underlying network may only be available by calling some application programming interfaces (APIs),
which support the function to retrieve a list of user's
friends.  Graph sampling through crawling is widely used in this scenario to
estimate graph properties such as degree distribution~\cite{lee2012beyond,
lirandom, gjoka2010walking}, clustering
coefficient~\cite{hardiman2013estimating} and size of
graphs~\cite{katzir2011estimating}. 
In this work, we assume that the graph data has to be externally accessed, either
through remote databases or by calling APIs provided by the operators of OSNs.

\textit{The aim of this work is to design and implement efficient random
walk-based methods to estimate graphlet concentration for restricted accessed
graphs}.  Note that estimating graphlet concentration is a more complicated task
than estimating other graph properties such as degree distribution. For degree
distribution, one can randomly walk on the graph to collect node samples and
then remove the bias using standard techniques such as Horvitz-Thompson
estimator.
For graphlet concentration, one needs to consider the various local structures.  A
single node sample {\em cannot} tell us information about the local structures.
One needs to map a random walk to a Markov chain, and carefully define the
state space and its transition matrix
so as to ensure that the state space contains all the $k$-node graphlets. 

\subsection{Related Works and Existing Problems.} 
Previous studies on graphlet counts or concentration include sampling methods
for (1) memory-based graphs~\cite{seshadhri2013triadic, jha2015path,
DBLP:journals/corr/WangTZG15, RahmanGAG}, (2) streaming
graphs~\cite{wang2015minfer, Ahmed:2014:GSH:2623330.2623757}, and (3)
restricted accessed graphs~\cite{hardiman2013estimating, bhuiyan2012guise,
wang2014efficiently}.  The state-of-the-art sampling methods for memory-based
graphs are wedge sampling~\cite{seshadhri2013triadic} and path
sampling~\cite{jha2015path}.  Wedge sampling in~\cite{seshadhri2013triadic}
generates uniformly random wedges (i.e., paths of length two) from the graphs to
estimate triadic measures (e.g., number of triangles, clustering coefficient).
Later on, Jha et al.~\cite{jha2015path} extended the idea of wedge sampling and
proposed path sampling to estimate the number of $4$-node graphlets in the
graphs.  However, both of wedge sampling and path sampling need to access the
whole graph data, which renders them impractical for restricted accessed graphs.  

Estimating graphlet counts for streaming graphs has been studied
in~\cite{wang2015minfer, Ahmed:2014:GSH:2623330.2623757}.  Ahmed et
al.~\cite{Ahmed:2014:GSH:2623330.2623757} proposed the {\em graph sample and
hold} method to estimate the triangle counts in the graphs. Wang et
al.~\cite{wang2015minfer} proposed a method to infer the number of any $k$-node
graphlets in the graph with a set of uniformly sampled edges from the graph
stream.  These streaming sampling methods access each edge at least once and are
not applicable to restricted accessed graphs.

Most relevant to our framework are previous random walk-based 
methods~\cite{hardiman2013estimating, bhuiyan2012guise, wang2014efficiently}
designed for graphs with restricted access.  
Hardiman and Katzir\cite{hardiman2013estimating} proposed a
random walk-based method to estimate the clustering coefficient, which
is a variant of the $3$-node graphlet concentration we study.
Bhuiyan et al.~\cite{bhuiyan2012guise} proposed \textit{GUISE}, which is based
on the Metropolis-Hasting random walk on a subgraph relationship graph, whose
nodes are all the $3, 4, 5$-node subgraphs.  They aimed to estimate $3,
4, 5$-node graphlet statistics simultaneously,  but \textit{GUISE} suffers
from rejection of samples.  In~\cite{wang2014efficiently}, authors proposed
three random walk-based methods: subgraph random walk (\textit{SRW}), pairwise
subgraph random walk (\textit{PSRW}), and mix subgraph sampling (\textit{MSS}).  
\textit{MSS} is an extension of \textit{PSRW} to estimate $k - 1, k, k +
1$-node graphlets jointly.  The simulation results show that \textit{PSRW}
outperforms \textit{SRW} in estimation accuracy. To the best of our knowledge,
\textit{PSRW} is the state-of-the-art random walk-based method to estimate
graphlet statistics  for restricted graphs.

We denote the subgraph relationship graph in~\cite{wang2014efficiently} as
$G^{(d)}$, and each element in $G^{(d)}$ is a $d$-node connected subgraphs in
the original graph.  The main idea of \textit{PSRW} is to collect $k$-node
graphlet samples through two consecutive steps of a simple random walk on
$G^{(k-1)}$ to the estimate $k$-node graphlet concentration.  One drawback of
\textit{PSRW} is its inefficiency of choosing neighbors during the random walk.
For example, \textit{PSRW} performs the random walk on $G^{(3)}$ to estimate
$4$-node graphlet concentration. Populating neighbors of nodes in $G^{(3)}$ is
about \textit{an order of magnitude} slower than choosing random neighbors of
nodes in $G^{(2)}$.  If one can figure out how to estimate $4$-node graphlet
concentration with random walks on $G^{(2)}$, the time cost can be reduced
dramatically.  Furthermore, since \textit{PSRW} is more accurate than the
simple random walk on $G^{(k)}$ (\textit{SRW}) when estimating $k$-node
graphlet concentration, we have reasons to believe that random walks on
$G^{(d)}$ with smaller $d$ have the potential to achieve higher accuracy.
Faster random walks and more accurate estimation motivate us to propose more
efficient sampling methods based on random walks on $G^{(d)}$ to estimate
$k$-node graphlet concentration. Different from \textit{PSRW}, we seek for $d$
that is smaller than $k-1$. 

\subsection{Our Contributions}
\noindent\textbf{Novel framework.} In this paper, we propose a novel framework
to estimate the graphlet concentration. Our framework is {\em provably
correct} and makes no assumption on the graph structures.  The main idea of
our framework is to collect samples through consecutive steps of a random walk
on $G^{(d)}$ to estimate $k$-node graphlet concentration, here $d$ can be {\em any
positive integer less than $k$}, and \textit{PSRW} is just a special case where
$d = k - 1$.  We construct the subgraph relationship graph $G^{(d)}$ on the
fly, and we do not need to know the topology of the original graph in advance.
In fact, one can view $d$ as a parameter of our framework.  As mentioned
in~\cite{wang2014efficiently}, {\em it is non-trivial to analyze and remove the
sampling bias when randomly walking on $G^{(d)}$ where $d$ is less than $k -
1$}.  The analysis method in \textit{PSRW} cannot be applied to the situation
where $d < k - 1$. Our work is not a simple extension of \textit{PSRW}. More
precisely, we propose a new and general framework which subsumes \textit{PSRW}
as a special case.  When choosing the appropriate parameter $d$, our methods
significantly outperform the state-of-the-art methods.

\noindent\textbf{Efficient optimization techniques.} We also introduce two
novel optimization techniques to further improve the efficiency of our framework.
The first one, \textit{corresponding state sampling (CSS)}, modifies the
re-weight coefficient and improves the efficiency of our estimator. The second
technique integrates the non-backtracking random walk in our framework. 
Simulation results show that our optimization techniques can improve the
estimation accuracy.   

\noindent\textbf{Provable guarantees.} We give detailed theoretical analysis on
our unbiased estimators. Specifically, we derive an analytic Chernoff-Hoeffding
bound on the sample size.  The theoretical bound guarantees the convergence of
our methods and provides insight on the factors which affect the performance of
our framework.

\noindent\textbf{Extensive experimental evaluation.} To further validate our
framework, we conduct extensive experiments on real-world networks.
In Section~\ref{sec:experiment}, we
demonstrate that our framework with an appropriate chosen parameter $d$ is more
accurate than the state-of-the-art methods. For $3$-node graphlets, our method with
the random walk on $G$ outperforms \textit{PSRW} up to 3.8$\times$ in accuracy. For
$4, 5$-node graphlets, our methods outperform
\textit{PSRW} up to 10$\times$ in accuracy and 100$\times$ in time cost.  
In summary:
\begin{itemize}[leftmargin=*]
    \item We propose a general Markov Chain and Monte Carlo
        (MCMC) framework to estimate the graphlet concentration.
    \item We derive an unbiased estimator for the framework
        and develop a Chernoff-Hoeffding bound for the sample size.
    \item We introduce two novel optimization techniques to improve the
        accuracy. The simulation results show that the techniques improve
        the estimation efficiency significantly.
    \item We conduct extensive experimental evaluation to 
        support our theoretical arguments.
\end{itemize}

The remainder of this paper is organized as follows.
Section~\ref{sec:preliminary} provides preliminaries.
Section~\ref{sec:framework} explains our framework in detail.
Section~\ref{sec:improvement} presents our two optimization techniques.
Section~\ref{sec:implementation} gives some implementation details.
Section~\ref{sec:experiment} reports our experimental results,
and finally, conclusion is given in 
Section~\ref{sec:conclusion}.

\section{Preliminary}\label{sec:preliminary}
In this section, we first define some notations and concepts used throughout
the paper. Then we review a basic Markov chain theory which serves as the
mathematical foundation of our random walk-based framework.

\subsection{Notations and Definitions}
Networks can be modeled as a graph $G=(V, E)$, where $V$ is the set of nodes
and $E$ is the set of edges. For a node $v\in V$, $d_v$ denotes the degree of
node $v$, i.e., the number of neighbors of node $v$.  A graph with neither
self-loops nor multiple edges is defined as a simple graph. In this work, we
consider \textit{simple}, \textit{connected} and \textit{undirected} graphs. 

\noindent \textbf{Induced subgraph.} A $k$-node induced subgraph is a subgraph
$G_k=(V_k, E_k)$ which has $k$ nodes in $V$ together with any edges whose both
endpoints are in $V_k$. Formally, we have $V_k\subset V$, $|V_k|=k$ and
$E_k\!=\!\{(u, v):\ u, v\in V_k\wedge (u, v)\in E\}$. 

\noindent \textbf{Subgraph relationship graph.} In~\cite{bhuiyan2012guise,
wang2014efficiently}, the authors proposed the concept of subgraph relationship
graph. Here we adopt the definition in~\cite{wang2014efficiently} and define
the $d$-node subgraph relationship graph $G^{(d)}$ as follows. Let $H^{(d)}$
denote the set of all $d$-node \textit{connected induced} subgraphs of $G$.
For $s_i, s_j \in H^{(d)}$, there is an edge between $s_i$ and $s_j$ if and
only if they share $d-1$ common nodes in $G$. We use $R^{(d)}$ to denote the
set of edges among all elements in $H^{(d)}$. Then we define $G^{(d)} =
(H^{(d)}, R^{(d)})$.  Specially, we define $G^{(1)} = G, H^{(1)} = V, R^{(1)} =
E$ when $d = 1$. If the original graph $G$ is connected, then $G^{(d)}$ is also
connected~\cite[Theorem 3.1]{wang2014efficiently}.
Figure~\ref{fig:subgraph_relationship} shows an example of $G^{(2)}$ and
$G^{(3)}$ for a $4$-node graph $G$.  Let $H^{(2)}$ denote all 2-node induced
subgraphs of $G$, then the node set of $G^{(2)}$ is $H^{(2)}$, i.e., node
pairs $\{(1,2), (1,3), (1,4), (2,3), (3,4)\}$.  Note that there is an edge
between node pair $(1,2)$ and $(2,3)$ in $G^{(2)}$ because they share node $2$
in $G$.

\begin{figure}[htbp]
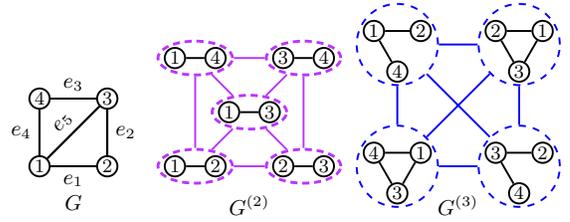
    
    \centering
    \scalebox{0.9}{
    \SUBGRAPHEXAMPLE}
    \caption{Original graph $G$ and its $2$ and $3$-node subgraph relationship graph
    $G^{(2)}$ and $G^{(3)}$.}
    \label{fig:subgraph_relationship}
\end{figure}
In general, constructing $G^{(d)}$ is impractical due to intensive computation
cost.  However, for our random walk-based framework, there is {\em no need} to
construct $G^{(d)}$ in advance since we can generate the neighborhood subgraphs
of $s\in H^{(d)}$ on the fly according to the definition of $G^{(d)}$.

\noindent \textbf{Isomorphic.} Two graphs $G=(V, E)$ and $G^\prime=(V^\prime,
E^\prime)$ are isomorphic if there exists a bijection $\varphi: V\rightarrow
V^{\prime}$ with $(v_i, v_j)\in E \Leftrightarrow (\varphi(v_i),
\varphi(v_j))\in E^{\prime}$ for all $v_i, v_j\in V$. Such a bijection map is
called an \textit{isomorphism}, and we write isomorphic $G$ and $G^\prime$ as
$G\simeq G^\prime$.

\begin{definition}
    Graphlets are formally defined as connected, \textit{non-isomorphic},
    \textit{induced subgraphs} of a graph $G$.
\end{definition}
Figure~\ref{fig:graphletexample} shows all $3, 4$-node graphlets. There are $2$
different $3$-node graphlets and $6$ different $4$-node graphlets. The second
row of Table~\ref{table:fivenodegraphlet} shows $21$ different $5$-node
graphlets. The number of distinct graphlets grows exponentially with the number
of vertices in the graphlets. For example, there are $112$ different $6$-node
graphlets and $853$ different $7$-node graphlets.  Due to the combinatorial
complexity, the computation for graphlets is usually restricted to $3, 4, 5$
nodes~\cite{jha2015path, DBLP:journals/corr/WangTZG15, ahmed2015icdm,
wang2014efficiently, bhuiyan2012guise, wang2015minfer,
hovcevar2014combinatorial}. Note that various applications, e.g.,~\cite{5664,
UganderSFM}, focus on graphlets with no more than $5$ nodes since graphlets
with up to $5$ nodes have the best cost-benefit trade
off~\cite{bhuiyan2012guise}.
\begin{figure}[htb]
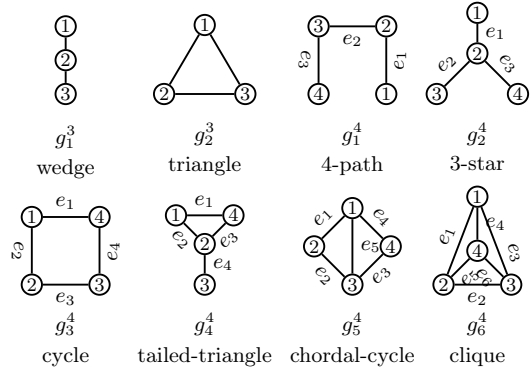

    \centering
    \scalebox{0.9}{
        \MOTIFFOURSUB{1}}
\caption{All $3, 4$-node distinct graphlets.}\label{fig:graphletexample}
\end{figure}

\noindent \textbf{Problem definition.} 
Given an undirected connected graph $G$ and all the distinct $k$-node graphlets
$\mathcal{G}^k = \{g^k_1, g^k_2, \cdots, g^k_m\}$,
where $g_{i}^k$ is the $i^{th}$ type
of $k$-node graphlets. 
Let $C^k_i$ denote the number of induced subgraphs that are
isomorphic to graphlet $g^k_i$. Our goal is to estimate the concentration of
$g^k_i\in \mathcal{G}^k$ for $G$, where the
concentration of $g^k_i$ is defined as
\begin{equation}\label{eq:concentration_def}
    c^k_i \triangleq \frac{C^k_i}{\sum^{m}_{j=1}C^k_j}.
\end{equation}
{\bf \em Example:} The graph $G$ in Figure~\ref{fig:subgraph_relationship} has two
triangles (induced by $\{1, 3, 4\}$ and $\{1, 2, 3\}$ respectively) and two
wedges (induced by $\{4, 1, 2\}$ and $\{2, 3, 4\}$ respectively). Then the
wedge concentration $c^3_1 = 0.5$ and the triangle concentration $c^3_2 = 0.5$ in $G$.

\noindent \textbf{Applications.} In the following, we list some important applications
involving graphlet concentration.
\begin{itemize}[leftmargin=*]
    \setlength\itemsep{0em}
    \item {\em Clustering coefficient.} Friends of friends tend to become
        friends themselves in OSNs. This property is referred to as
        ``transitivity''. Clustering coefficient, which is defined as
        $3C^3_2/(C^3_1 + 3C^3_2) = 3c^3_2/ (2c^3_2 + 1)$, quantifies the
        transitivity of networks, i.e., the probability that two neighbors of
        any vertex are connected. The clustering coefficient is important to
        understand the networks and can be obtained directly with the triangle
        concentration $c^3_2$. 
    \item {\em Large graph comparison and classification.} 
        One can use the graphlet concentration
        as the fingerprint for graph comparison~\cite{ahmed2015icdm}. 
        The $3, 4, 5$-node graphlet concentration was
        proposed as the features for large graph classification
        in~\cite{5664}. The intuition of using the graphlet concentration
        instead of the graphlet counts 
        is that the differences in sizes of
        graphs skew the graphlet counts greatly and may result in poor
        performance of graph classification.
    \item {\em Intrinsic properties analysis.} Graphlet concentration can be
        used to understand the intrinsic properties of the networks. For
        example, Ahmed et al.~\cite{ahmed2015icdm} computed the $4$-node
        graphlet concentration of the OSN Friendster and found that Friendster is
        lack of community related graphlets (e.g., cliques), which indicates
        the collapse of Friendster. Wang et al.  in~\cite{wang2014efficiently}
        used the concentration of directed $3$-node graphlets to analyze the
        differences in functions of two OSNs, i.e., Douban and Sinaweibo. 
\end{itemize}
Computing graphlet concentration is by no way an easier task than computing
graphlet counts. Actually, graphlet counts are just reflections of the sizes of
graphs. In Section~\ref{subsec:unbiased_estimator} we show that the graphlet
counts can be reconstructed easily if we have access to the whole graph data.

\subsection{Random Walk and Markov Chain}\label{subsec:markov_chain}
Our framework is to generate samples from random walks. A \textit{simple random
walk} (SRW) over graph $G$ is defined as follows: start from an initial node
$v_0$ in $G$, we move to one of its neighbors which is chosen uniformly at
random, and repeat this process until certain stopping criteria. The random
walk on graph $G$ can be viewed as a \textit{finite} and
\textit{time-reversible} Markov chain with state space $V$. More specifically,
let $\{X_t\}$ be the Markov chain representing the visited nodes of the random
walk on $G$, the transition probability matrix $\mathbf{P}$ of this Markov
chain is defined as
\begin{equation*}
    \mathrm{P}(i, j) = 
    \begin{cases}
        \frac{1}{d_i}, & \text{if }(i, j)\in E,\\
        0, & \text{otherwise.}
    \end{cases}
\end{equation*}
The SRW has a \textit{unique} stationary distribution $\pmb{\pi}$ where
$\pi(v)=\frac{d_v}{2|E|}$ for $v\in V$ ~\cite{Lovasz1996, olle2000finite}.  The
stationary distribution is vital for
correcting bias of samples generated by random walks.  

\noindent\textbf{Strong Law of Large Numbers.}
Below, we review the {\em Strong Law of Large Numbers} (SLLN)
for Markov chain which serves as the theoretical foundation for our random
walk-based framework.  For a Markov chain with finite state space
$\mathcal{S}$ and stationary distribution $\pmb{\pi}$, we define the
expectation of the function $f: \mathcal{S}\rightarrow \mathbb{R}$ with
respect to $\pmb{\pi}$ as
\begin{equation*}
    \mu  = \mathbb{E}_{\pmb{\pi}}[f] \triangleq \sum_{X\in \mathcal{S}} f(X)\pi(X).
\end{equation*}
The subscript $\pmb{\pi}$ indicates that the expectation is calculated with
the assumption that $X\sim \pmb{\pi}$. Let 
\begin{equation*}
    \hat{\mu}_n = \frac{1}{n}\sum_{s=1}^{n}f(X_s)
\end{equation*}
denote the sample average of $f(X)$ over a run of the Markov chain.  The
following theorem SLLN ensures that the sample mean converges almost surely to its expected value.
\begin{theorem}
    \cite{geyer1998markov, lee2012beyond} For a finite and irreducible Markov
    chain with stationary distribution $\pmb{\pi}$, we have
    \begin{equation*}
        \hat{\mu}_n\ \rightarrow \mathbb{E}_{\pmb{\pi}}[f]\ \text{almost surely
        (a.s.)}
    \end{equation*}
    as $n\rightarrow \infty$ regardless of the initial distribution of the chain.
\end{theorem}
The SLLN is the fundamental basis for most random walk-based graph sampling
methods, or more formally, Markov Chain Monte Carlo (MCMC) samplers,
e.g.,~\cite{lirandom,bhuiyan2012guise, hardiman2013estimating}. The SLLN
guarantees the asymptotic unbiasedness of the estimators based on any finite
and irreducible Markov chain.

\noindent{\bf Mixing Time.}
The mixing time of a Markov chain is the number of steps it takes for a random
walk to approach its stationary distribution. We adopt the definition of mixing
time in~\cite{Mitzenmacher, mohaisen2010measuring, chung2012chernoff}. The
mixing time is defined as follows.
\begin{definition}
    The mixing time $\tau(\epsilon)$ (parameterized by $\epsilon$) of a Markov
    chain is defined as
\begin{equation*}
    \tau(\epsilon) =\max_{X_i\in \mathcal{S}} \min\{t: |\pmb{\pi} -
        \mathbf{P}^t\pmb{\pi}^{(i)}|_1 < \epsilon\},
\end{equation*}
where $\pmb{\pi}$ is the stationary distribution of the Markov chain,
$\pmb{\pi}^{(i)}$ is the initial distribution when starting from state
$X_i\in \mathcal{S}$, $\mathbf{P}^{t}$ is the transition matrix after $t$
steps and $|\cdot|$ is the variation distance between two
distributions\footnote{The variation
distance between two distributions $\mathbf{d}_1$ and $\mathbf{d}_2$ on a
countable state space $\mathcal{S}$ is given by $|\mathbf{d}_1 -
\mathbf{d}_2|_1\triangleq\frac{1}{2}\sum_{x\in \mathcal{S}}|d_1(x) -
d_2(x)|$.
}. 
\end{definition}
Later on, we will use the mixing time based Chernoff-Hoeffding
bound~\cite{chung2012chernoff} to compute the needed sample size to guarantee that
our estimate is within $(1\pm \epsilon)$ of the true value with probability at
least $1-\delta$. 

\section{General Framework}\label{sec:framework}
In this section, we introduce our random walk-based general framework for
graphlet concentration estimation. 
Our framework leverages ``{\em consecutive
steps}'' of the random walk.  We derive an unbiased estimator for graphlet
concentration with the {\em re-weighting techniques} for SRW. With this
unbiased estimator, we can estimate, in general, any $k$-node graphlet
concentration.  We will illustrate how to estimate $k=3, 4, 5$-node graphlet
concentration in particular.
For ease of presentation, we summary the notations in
Table~\ref{table:notations}.
\begin{table}[t]
    \caption{Summary of Notations}\label{table:notations}
    \setlength{\tabcolsep}{1pt}
    \resizebox{0.48\textwidth}{!}{
    \begin{tabular}{|c|c|}
        \hline
        $G$ & $G=(V, E)$, underlying undirected graph\\ 
        $G^{(d)}$& {\small $G^{(d)}\!=\!(H^{(d)},\!R^{(d)})$}, 
        $d$-node subgraph relationship graph\\
        $g^k_i$ & the $i^{th}$ type of $k$-node graphlets\\
        $C^k_i$ & number of subgraphs isomorphic to $g^k_i$ in graph $G$\\
        $c^k_i$ & concentration of graphlet $g^k_i$ in $G$ \\ 
        $d_v$ & degree of node $v$ in $G$\\
        $d_X$ & degree of state $X$, $X$ is also a node in $G^{(d)}$\\
        $\mathcal{M}^{(l)}$ & state space of the expanded Markov chain \\
        $X^{(l)}$ & state in state space $\mathcal{M}^{(l)}$\\
        $V(X^{(l)})$ & set of graph $G$'s nodes contained in state $X^{(l)}$ \\ 
        $s(X^{(l)})$ & subgraph induced by node set $V(X^{(l)})$ \\ 
        $\pmb{\pi}$ & stationary distribution of the random walk \\
        $\PIMB$ & stationary distribution of the expanded Markov chain \\
        $\alpha^k_i$ & number of states $X^{(l)}$ in $\mathcal{M}^{(l)}$ s.t.
        $s\simeq g^k_i$, $s(X^{(l)}) = s$.\\
        $\mathcal{C}(s)$ & set of states $X^{(l)}$ in $\mathcal{M}^{(l)}$ s.t.
        $s(X^{(l)})  = s$.\\ 
        \hline
    \end{tabular}
    }
\end{table}

\subsection{Basic Idea}\label{subsec:basic_idea_of_framework}
We are interested in finding the concentration of $g^k_i$, where $k \geq 3$
(since it is not difficult to find for $k=1$ or $2$).  The core idea is to
collect $k$-node graphlet samples through $l \!=\! k - d + 1$ consecutive steps
of the random walk on $G^{(d)}$, where $d \in \{1,\ldots,k-1\}$. 
Specifically, for each state $X_i$ ($i \geq l$) of the random walk on
$G^{(d)}$, suppose we keep the history of previously visited $l-1$ states, then we
consider the subgraph induced by nodes contained in $X^{(l)} = (X_{i - l +
1},\ldots, X_i)$ as a $k$-node graphlet sample.  Note that we only consider
consecutive $l$ steps which visit $k$ distinct nodes in $V$. If the consecutive
$l$ steps on $G^{(d)}$ fail to collect $k$ distinct nodes in $V$, we just
continue the random walk until we find $l$ consecutive steps which contain $k$
distinct nodes in $V$. 

\noindent{\bf \em Example:} Consider the graph in Figure~\ref{fig:subgraph_relationship}. 
    (a) Suppose we want to get $3$-node graphlet samples and choose a random
        walk on $G$. Then $l = 3 -1 + 1 = 3$, i.e., we need to walk for $3$
        steps on $G$ to get the $3$-node graphlet samples.  Assume the random
        walk sequence is $1\rightarrow 2\rightarrow 1\rightarrow 4\rightarrow
        3$. Then we can obtain two $3$-node graphlet samples induced by $\{2,
        1, 4\}$ and $\{1, 4, 3\}$ respectively. The sequence $1\rightarrow
        2\rightarrow 1$ is discarded since it only visits two distinct nodes. 
    (b) If we want to explore possible graphlets $g^4_i$ and have decided to
        perform a random walk on $G^{(2)}$ (i.e., $d=2$).  Assume we make
        $l=3$ transitions on the following states: $(1, 2)\rightarrow(1,
        3)\rightarrow(3,4)$, then we can obtain a $4$-node graphlet sample
        induced by the node set $\{1, 2, 3, 4\}$, 
        because $\{1, 2, 3, 4\}$ is contained in the three states $(1, 2),
        (1, 3)$ and $(3, 4)$. In this case, the
        obtained graphlet sample corresponds to $g_{5}^{4}$ in 
        Figure~\ref{fig:graphletexample}.

The main technical challenge is to remove the bias of the obtained graphlet
samples. To analyze the bias theoretically, we first introduce the concept of
\textit{expanded Markov chain}.

\subsection{Expanded Markov Chain}\label{subsec:expanded_markovchain}
We define an \textit{expanded Markov chain} which remembers consecutive $l$
steps of the random walk on $G^{(d)}= (H^{(d)}, R^{(d)})$. Each $l$ consecutive
steps are considered as a state $X^{(l)} = (X_1, \cdots, X_l)$ of the expanded
Markov chain. Here we use the superscript ``$l$'' to denote the length of the
random walk block. Each time the random walk on $G^{(d)}$ makes a transition,
the expanded Markov chain transits to the next state.  Assume the expanded
Markov chain is currently at state $X^{(l)}_i=(X_1,\cdots, X_l)$, it means the
random walker is at $X_l$. If the walker jumps to $X_{(l+1)}$, i.e., one of the
neighbors of $X_l$, then the expanded Markov chain transits to the state
$X^{(l)}_{i+1}=(X_2, \cdots, X_{l+1})$.  
Let $\mathcal{N} = H^{(d)}$ denote the state space for
the random walk and $\mathcal{M}^{(l)}$ denote the state space for the
corresponding expanded Markov chain.  The state space $\mathcal{M}^{(l)}$
consists of all possible consecutive $l$ steps of the random walk. More
formally, $\mathcal{M}^{(l)}=\{(X_1, \cdots, X_l): X_i\in \mathcal{N}, 1\leq
i\leq l{\ s.t.\ }(X_i, X_{i+1})\in R^{(d)}\ \forall 1\leq i\leq l-1\}\subseteq
\mathcal{N}\times\cdots\times\mathcal{N}$.  For example, if we perform a random
walk on $G$, any $(u, v)$ and $(v, u)$ where $e_{uv}\in E$ are states in
$\mathcal{M}^{(2)}$. 
Note that the expanded Markov chain describes the same process as the random
walk. The reason we define it here is for the convenience of deriving unbiased
estimator.  

The bias caused by the non-uniform sampling probabilities of the graphlet
samples arises from two aspects.  First, the states in the expanded Markov
chain do not have equal stationary probabilities. Second, a graphlet sample
corresponds to several states in the expanded Markov chain. To derive the
unbiased estimator of the graphlet concentration, we compute the stationary
distribution of the expanded Markov chain and the number of states
corresponding to the same graphlet sample.

\noindent{\bf Stationary distribution.}
Let $\pmb{\pi}$ denote the stationary distribution of the random walk and $\PIMB$ denote the stationary distribution of the expanded Markov chain.  $d_{X}$ is
number of neighbors of state $X$ in $G^{(d)}$.  The following theorem states
that $\PIMB$ is unique and its value can be computed with $\pmb{\pi}$.

\begin{restatable}[]{theorem}{stationarydistribution}
    \label{thm:stationary_distribution}
    The stationary distribution $\PIMB$ exists and is unique. 
    For any $X^{(l)} = (X_1, \cdots, X_l)\in \mathcal{M}^{(l)}$, we have 
    \begin{equation}\label{eq:stationary_distribution}
        \begin{split}
            \pi_{e}(X^{(l)}) =\begin{cases}
                \frac{d_{X_l}}{2|R^{(d)}|} &\text{if } l = 1,\\
            \frac{1}{2|R^{(d)}|} &\text{if } l = 2,\\
            \frac{1}{2|R^{(d)}|}\frac{1}{d_{X_2}}\cdots\frac{1}{d_{X_{l - 1}}} 
            &\text{if } l > 2
            \end{cases}
        \end{split}
    \end{equation}
\end{restatable}

The proof of Theorem~\ref{thm:stationary_distribution}
is in Appendix~\ref{app:section:stationary_distribution}.
In fact, $\PIMB$ can be derived directly using conditional probability
formula. For the state $X^{(l)} = (X_1, \cdots, X_l)$, $X_1$ is visited with
probability $\pi(X_1) = \frac{d_{X_1}}{2|R^{(d)}|}$ during the random walk and
$\pi_e(X^{(l)})$ can be written as $\frac{d_{X_1}}{2|R^{(d)}|}\times
\frac{1}{d_{X_1}}\times \cdots\times \frac{1}{d_{X_{l}}}$.

\noindent{\bf \em Example:}
Still using the graph in Figure~\ref{fig:subgraph_relationship} as an example,
if we walk on $G^{(2)}$ and visit states $X_1 = (1, 2), X_2 = (1, 3), X_3 = (3,
4)$, then the corresponding state $X^{(3)}_1$ in $\mathcal{M}^{(3)}$ is $(X_1,
X_2, X_3)$. The number of edges in $G^{(2)}$ is 8.  The degrees of $X_1, X_2,
X_3$ are $3, 4, 3$ respectively. Then the stationary distribution of
$X^{(3)}_1$ is $1/16\cdot1/4=1/64$. 

\noindent{\bf State corresponding coefficient.} Let $V(X^{(l)})$ represent the
set of graph $G$'s nodes contained in the state $X^{(l)}$, and $s(X^{(l)})$ be
the subgraph induced by $V(X^{(l)})$. We define $X^{(l)}$ as the
\textit{corresponding state} for the subgraph $s(X^{(l)})$. A key observation
is that a subgraph may have several corresponding states in
$\mathcal{M}^{(l)}$. For example, if we perform a random walk on $G$, then a
triangle induced by $\{u, v, w\}$ in $G$ has 6 corresponding states $(u, v,
w)$, $(u, w, v)$, $(v, u, w)$, $(w, v, u)$, $(v, w, u)$, $(w, u, v)$ in
$\mathcal{M}^{(3)}$. To describe this idea formally, we define the
\textit{state corresponding coefficient} $\alpha^k_i$.
\begin{definition}
    For any connected induced subgraph $s\simeq g^k_i$, we define the set of
    corresponding states for $s$ as $\mathcal{C}(s)=\{X^{(l)}|s(X^{(l)}) = s,
    X^{(l)}\in \mathcal{M}^{(l)}\}$ and state corresponding coefficient
    $\alpha^k_i$ as $|\mathcal{C}(s)|$.  
\end{definition}
Coefficient $\alpha^k_i$ is vital to the design of unbiased estimator. 
If states in $\mathcal{M}^{(l)}$ are
uniformly sampled, then the probability of getting a subgraph isomorphic to $g^k_i$
is $\alpha^k_i C^k_i/|\mathcal{M}^{(l)}|$.  The physical interpretation of
$\alpha^k_i$ is that each subgraph isomorphic to $g^k_i$ is {\em replicated}
$\alpha^k_i$ times in the state space. If $\alpha^k_i$ is larger, we have a
higher chance to get samples of $g^k_i$.  Note that coefficient $\alpha^k_i$
only depends on different graphlets and random walk types. Hence we can compute
$\alpha^k_i$ in advance.
For example, if we perform random walk on $G$ and $l = 3$, we have $\alpha^3_2
= 6$, i.e., any triangle in $G$ has $6$ corresponding states in
$\mathcal{M}^{(3)}$.  In another words, we have $6$ ways to traverse the
triangle $\{u, v, w\}$ with random walk on $G$.

In fact, $\alpha^k_i$ denotes how many ways we can traverse the
$g^k_i$.
Theoretically, $\alpha^k_i$ equals to {\em twice} of the number of Hamilton
paths\footnote{Hamilton path is defined as a
path that visits each vertices in the graph exactly once.}
(each Hamilton path is counted from both directions) in the subgraph
relationship graph of $g^k_i$.  The
detailed computation process of $\alpha^k_i$ is presented in the
Appendix~\ref{app:computation_of_coefficient}. Table~\ref{table:threefournode}
lists the coefficient $\alpha^k_i$ for $k = 3, 4$. $SRW(d)$ in the table
represents a random walk on $G^{(d)}$. Notice that $\alpha^4_2 = 0$ if we
choose $SRW(1)$, i.e., if we walk on $G$, there is no chance to get samples of
$g^4_2$. In this case, we can only estimate the relative concentration for all
$4$-node graphlets except $g^4_2$ ($3$-star)\footnote{This problem can be solved
by using linear equations between counts of non-induced subgraphs and induced subgraphs.
The details are omitted for clear presentation and limitation of space.}.
Table~\ref{table:fivenodegraphlet} lists the coefficient for all $5$-node
graphlets. 

\begin{table}[t]
        \caption{Coefficient $\alpha^k_i$ for $k = 3, 4$ nodes
        graphlets.}\label{table:threefournode}
        \begin{tabular}{|*{11}{c|}}
        \hline
        \multicolumn{2}{|c|}{Graphlet}     
        & $g^3_1$   & $g^3_2$  & $g^4_1$   & $g^4_2$   & $g^4_3$ & $g^4_4$  &
        $g^4_5$ & $g^5_6$ \\ \hline
        \multirow{3}{*}{$\alpha^k_i / 2$} & $SRW(1)$     
        & $1$   & $3$  & $1$   & $\mathbf{0}$   & $4$ & $2$ & $6$   & $12$  
        \\ \cline {2-10} 
        & $SRW(2)$    
        & $1$   & $3$  & $1$   & $3$   & $4$ & $5$ & $12$   & $24$   
        \\ \cline{2-10}
        & $SRW(3)$ 
        & $1/2$   & $1/2$  & $1$   & $3$   & $6$ & $3$ & $6$   & $6$   
        \\ \hline
    \end{tabular}
\end{table}

\begin{table*}[t]
    \begin{minipage}{\textwidth}
        \caption{Coefficient $\alpha^5_i$ for $5$-node
        graphlets}\label{table:fivenodegraphlet}
        \centering
        \footnotesize\setlength{\tabcolsep}{1pt}
        \resizebox{\textwidth}{!}{
        \begin{tabular}{|*{23}{c|}}
        \hline
        \multicolumn{2}{|c|}{ID}     
        & $1$   & $2$  & $3$   & $4$   & $5$ & $6$  & $7$ & $8$ & $9$ 
        & $10$ & $11$  & $12$   & $13$   & $14$ & $15$  & $16$  & $17$ 
        & $18$ & $19$ & $20$ & $21$
        \\ \hline
        \multicolumn{2}{|c|}{Shape}     
        & \adjustbox{valign=m}{\GFIVE{1}}  & \adjustbox{valign=m}{\GFIVE{2}} 
        & \adjustbox{valign=m}{\GFIVE{3}}  & \adjustbox{valign=m}{\GFIVE{4}}
        & \adjustbox{valign=m}{\GFIVE{5}}  & \adjustbox{valign=m}{\GFIVE{6}} 
        & \adjustbox{valign=m}{\GFIVE{7}}  & \adjustbox{valign=m}{\GFIVE{8}}
        & \adjustbox{valign=m}{\GFIVE{9}}  
        & \adjustbox{valign=m}{\GFIVE{10}} 
        & \adjustbox{valign=m}{\GFIVE{11}} & \adjustbox{valign=m}{\GFIVE{12}}
        & \adjustbox{valign=m}{\GFIVE{13}} & \adjustbox{valign=m}{\GFIVE{14}}
        & \adjustbox{valign=m}{\GFIVE{15}} & \adjustbox{valign=m}{\GFIVE{16}}
        & \adjustbox{valign=m}{\GFIVE{17}} & \adjustbox{valign=m}{\GFIVE{18}}
        & \adjustbox{valign=m}{\GFIVE{19}} & \adjustbox{valign=m}{\GFIVE{20}}
        & \adjustbox{valign=m}{\GFIVE{21}} 
        \\ \hline
        \multirow{4}{*}{$\alpha^5_i / 2$} & $SRW(1)$ &1 &\textbf{0}
        &\textbf{0} & 1 & 2 & \textbf{0}
        &5 & 2 & 2 & 4 & 4 & 6 & 7 & 6 & 6 & 10 & 14 & 18 & 24 & 36 & 60
        \\ \cline{2-23} 
        & $SRW(2)$     
        & $1$   & $2$  & $12$   & $5$   & $4$ & $16$ & $5$   & $6$ & $24$ 
        & $24$ & $12$  & $18$  & $15$   &$54$  & $36$   & $42$ & $34$ & $82$ & $76$   & $144$  & $240$ 
        \\ \cline{2-23} 
        & $SRW(3)$    
        & $1$   & $5$  & $24$   & $8$   & $5$ & $24$ & $5$   & $16$ & $30$  
        & $24$ & $16$ & $63$   & $26$   & $63$ & $30$ & $43$   & $63$ & $63$ & $90$   & $90$  & $90$ 
        \\ \cline{2-23}
        & $SRW(4)$ 
        & $1$   & $3$  & $6$   & $3$   & $3$ & $6$ & $10$   & $12$ & $12$  
        & $12$ & $12$ & $10$   & $10$   & $10$ & $12$ & $10$   & $10$ & $10$  & $10$   & $10$  & $10$ 
        \\ \hline
    \end{tabular}
    }
    \end{minipage}
\end{table*}

\subsection{Unbiased Estimator}\label{subsec:unbiased_estimator}
Now we derive an unbiased estimator for the graphlet concentration.
Define an indicator function $h^k_i$ for the state $\X{l}$:
\begin{equation}\label{eq:indicator_function}
    h^k_i(X^{(l)}) = \mathbbm{1}\{s(X^{(l)})\simeq g^k_i\}.
\end{equation}
Function $h^k_i(X^{(l)}) = 0$ if the number of distinct nodes in $X^{(l)}$ is
less than $k$ or $s(X^{(l)})$ (subgraph induced by $G$'s nodes in $X^{(l)}$) is
not isomorphic to $g^k_i$.  Since each subgraph isomorphic to $g^k_i$ is {\em
replicated} $\alpha^k_i$ times in the state space, we have
\begin{equation*}
    \sum_{X^{(l)}\in \mathcal{M}^{(l)}} h^k_i(X^{(l)}) = \alpha^k_i C^k_i. 
\end{equation*}
Using above fact, we have
\begin{equation*}
    \small
    \mathbb{E}_{\PIMB}\left[\frac{h^k_i(X^{(l)})}{\pi_e(X^{(l)})}\right]
    =\sum_{X^{(l)}\in \mathcal{M}^{(l)}} \frac{h^k_i(X^{(l)})}{\pi_e(X^{(l)})}\pi_e(X^{(l)})
    = \alpha^k_i C^k_i.
\end{equation*}
Suppose there are $n$ samples $\{X^{(l)}_s\}^n_{s=1}$ obtained
from the expanded Markov chain. Combining the SLLN in
Section~\ref{subsec:markov_chain} and above equation, we have 
\begin{equation*}
    \hat{\mu}_n\triangleq\!\frac{1}{n}\!\sum_{s=1}^n\!\frac{h^k_i(X^{(l)}_s)}{\pi_e(X^{(l)}_s)} \rightarrow \alpha^k_i C^k_i.
\end{equation*}
Hence, we estimate $C^k_i$ as
\begin{equation}\label{eq:graphletcounter}
    \hat{C}^k_i \triangleq \frac{1}{n} \sum_{s =
1}^{n}\frac{h^k_i(X^{(l)}_s)}{\alpha^k_i\pi_e(X^{(l)}_s)}\rightarrow
    C^k_i.
\end{equation}
The bias of each graphlet sample induced by nodes in $X^{(l)}$ is corrected by
dividing its ``{\it inclusion probability}'' $\alpha^k_i\pi_e(X^{(l)})$.  Note
that this is a special case of \textit{importance sampling}~\cite[Chapter
9]{mcbook}.  With the {\em graphlet count estimator} in
Equation~\eqref{eq:graphletcounter}, we can derive the graphlet concentration
estimator easily.  Define $h^k(X^{(l)}) = \mathbbm{1}\{|V(X^{(l)})| = k\}$ and
$\alpha^k (X^{(l)}) = \sum_{i = 1}^{|\mathcal{G}^k|}\alpha^k_i h^k_i(X^{(l)})$.
The concentration $c^k_i$ is estimated with the following formula:
\begin{equation}\label{eq:concentration_estimator}
    \hat{c}^k_i \triangleq \frac{\sum_{s=1}^{n}h^k_i(X^{(l)}_{s}) /
    \left(\alpha^k_i\pi_e(X^{(l)}_{s})\right)}{\sum_{s=1}^{n} h^k(X^{(l)}_{s}) /
    \left(\alpha^k(X^{(l)}_{s})\pi_e(X^{(l)}_{s})\right)}\ \rightarrow\ c^k_i.
\end{equation}
\noindent {\bf \em Remarks:} The common denominator in $\pi_e(X^{(l)})$ is
$2|R^{(d)}|$, which is usually unknown for graphs with restricted access.
Fortunately, $|R^{(d)}|$ in the numerator and denominator of $\hat{c}^k_i$
cancels out. That means $c^k_i$ can be estimated without knowing $|R^{(d)}|$.
Local information (i.e., degree of nodes, adjacent relationship) collected along
the random walk is enough for the estimation. If we replace $\pi_e(X^{(l)})$
with $\tilde{\pi}_e(X^{(l)})\triangleq 2|R^{(d)}|\pi_e(X^{(l)})$ in
Equation~\eqref{eq:concentration_estimator}, the estimator $\hat{c}^k_i$ remains
unchanged. For the state $X^{(l)} = (X_1, \cdots, X_l)$,
$\tilde{\pi}_e(X^{(l)})$ can be computed directly with the degrees of $X_1, \cdots, X_l$.

Algorithm~\ref{algo:unbiased_estimate} depicts the process of $k$-node graphlet
concentration estimation.  Note that we can easily estimate the count of
graphlets with Equation~\eqref{eq:graphletcounter} if we know $|R^{(d)}|$.  For
random walk on $G$, we have $R^{(1)} = |E|$. For SRW on $G^{(2)}$,
$|R^{(2)}|=\frac{1}{2}\sum_{e_{uv}} (d_u + d_v - 2)$ and a single pass of graph
data is enough to compute this value.

\noindent{\bf Bound on sample size.}
The next important question we like to ask is what is the {\em smallest sample
size} (or random walk steps) to guarantee high accuracy of our estimator?  In
the following, we show the relationship between accuracy and sample size using
the Chernoff-Hoeffding bound for Markov chain.
Let $W$ denote $\max_{X^{(l)}}1 / \pi_e(X^{(l)})$ and
$\alpha_\text{min}=\min_i \alpha^k_i$. 
The following theorem states the
relationship between the estimation accuracy and sample size.
\begin{restatable}[]{theorem}{unionbound}\label{thm:union_bound}
For any $0 < \delta < 1$, there exists a constant
$\xi$ such that
\begin{equation}
    \Pr\left[(1 - \epsilon)c^k_i \leq \hat{c}^k_i \leq (1 +
    \epsilon)c^k_i\right] > 1 - \delta
\end{equation}
when the sample size $n \geq
\xi(\frac{W}{\Lambda})\frac{\tau}{\epsilon^2}(\log
\frac{\norm{\varphi}_{\PIMB}}{\delta})$. Here $\Lambda =
\min\{\alpha^k_iC^k_i, \alpha_\text{min}C^k\}$,
$\tau$ is the mixing time $\tau(1/8)$ of the original random walk, $\varphi$ is
the initial distribution and
$\norm{\varphi}_{\PIMB}$ is defined as $\sum_{X^{(l)}}\varphi^2(X^{(l)})/\pi_e(X^{(l)})$.
\end{restatable}

\noindent{\bf \em Remarks:} The proof for
Theorem~\ref{thm:union_bound} is in Appendix~\ref{app:chernoff_bound}.
From Theorem~\ref{thm:union_bound}, we know
that the needed sample size is {\em linear} with the mixing time $\tau$. This
implies our framework performs better for graphs with smaller mixing time.
Furthermore, some graphlet types are relatively rare in the graphs. For these
rare graphlet types, we need larger sample size to guarantee the same accuracy.
If $\alpha^k_i$ is higher for rare graphlet $g^k_i$, the needed sample size is
smaller. 

\begin{algorithm}[t]
\caption{Unbiased Estimate of Graphlet Statistics}
\begin{algorithmic}[1]\label{algo:unbiased_estimate}
    \REQUIRE sample budget $n$, SRW on $G^{(d)}$, graphlet size $k$
    \ENSURE estimate of $\left[c^k_1, \cdots, c^k_m\right]$ ($m =
    |\mathcal{G}^k|$)
    \STATE{random walk block length $l\leftarrow k - d + 1$ }\\
    \STATE{counter $\hat{C}^k_i\leftarrow 0$ for $i\in\{1,\cdots, m \}$}\\
    \STATE{walk $l$ steps to get the initial state $X^{(l)}\!=\!(X_1,\!\cdots,
    \!X_l)$}\\
    \STATE{random walk step $t\leftarrow 0$}
    \WHILE{$t < n$}
    \STATE{$i\leftarrow$ graphlet type id of subgraph $s(X^{(l)})$}
    \STATE{$\hat{C}^k_i\leftarrow \hat{C}^k_i +
    1/\left(\alpha^k_i\tilde{\pi}_e(X^{(l)})\right)$}
    \STATE{$X_{l + t + 1}\leftarrow$ uniformly chosen neighbor of $X_{l + t}$}
    \STATE{$X^{(l)}\leftarrow (X_{t + 2}, \cdots, X_{l + t + 1})$}
    \STATE{$t\leftarrow t + 1$}
    \ENDWHILE
    \STATE{$\hat{c}^k_i = \hat{C}^k_i/\sum_{j=1}^{m} \hat{C}^k_j$
    for all $i\in\{1,\cdots, m\}$}
    \STATE {\textbf{return} $\left[\hat{c}^k_1, \cdots, \hat{c}^k_m\right]$}
\end{algorithmic}
\end{algorithm}

\section{Improved Estimation}\label{sec:improvement}
We now introduce two novel optimization techniques to reduce the need sample
size in Theorem~\ref{thm:union_bound}.  The first technique is
to correct the bias by combining the stationary probabilities of the
corresponding states.  The second technique integrates the non-backtracking random
walk into our sampling framework.  With these two optimization techniques, we
obtain a more efficient estimator.

\subsection{Corresponding State Sampling (CSS)}
Recall that $\mathcal{C}(s)$ is defined as the set of states which correspond
to the subgraph $s$, i.e., states in $\mathcal{C}(s)$ contain the same set
of nodes as subgraph $s$.
The key observation is that for $X^{(l)}_a, X^{(l)}_b\in \mathcal{C}(s)$, the
``{\em inclusion probabilities}'' $\alpha^k_i\pi_e(X^{(l)}_a)$
and $\alpha^k_i\pi_e(X^{(l)}_b)$ may be different even though $X^{(l)}_a$ and
$X^{(l)}_b$ correspond to the same subgraph. In other words, the inclusion
probability of a subgraph is determined not only by the nodes in the subgraph,
but also the orders in which these nodes are visited.

\noindent{\bf \em Example:}
To illustrate, consider a triangle $\Delta$ induced by nodes $u, v, w$. Suppose
we choose a random walk on $G$. Then both of states $X^{(3)}_1 = (u, v, w)$ and
$X^{(3)}_2 = (v, u, w)$ correspond to the triangle $\Delta$.  We know that
$\alpha^3_2\pi_e(X^{(3)}_1)=\frac{6}{2|E|}\frac{1}{d_v}$ while
$\alpha^3_2\pi_e(X^{(3)}_2)=\frac{6}{2|E|}\frac{1}{d_u}$.  If nodes $u$ and $v$
have different degrees, then $\alpha^3_2\pi_e(X^{(3)}_1)\neq
\alpha^3_2\pi_e(X^{(3)}_2)$, i.e., the same triangle $\Delta$ has different
inclusion probabilities when visited in orders $u, v, w$ and $v, u, w$.   

Based on this observation\footnote{Here we require $l > 2$ since when $l = 2$,
the inclusion probabilities are the same for the states corresponding to the
same subgraph.}, we define ``{\it sampling probability}'' $p(X^{(l)})$ for
the subgraph induced by nodes in $X^{(l)}$. The value of $p(X^{(l)})$ only depends
degrees of nodes in $X^{(l)}$.
\begin{definition}
    For a state $X^{(l)}$ and a subgraph $s = s(X^{(l)})$, 
    we define the sampling probability for $s$ as
    \begin{equation*}
        p(X^{(l)}) \triangleq \sum_{X^{(l)}_j\in \mathcal{C}(s)} \pi_e(X^{(l)}_j).
    \end{equation*}
\end{definition}
In the following, we prove that if we substitute $\alpha^k_i\pi_e(X^{(l)})$ with
$p(\X{l})$ in Equation (\ref{eq:graphletcounter}), we still obtain an unbiased
estimator of $C^k_i$.

\begin{lemma}\label{lemma:equation_expectation}
For a specific subgraph $s\simeq g^k_i$, we have
\begin{equation*}
    \begin{split}
        \mathbb{E}_{\PIMB}[\frac{1}{\alpha^k_i\pi_e{X^{(l)}}}
     \mathbbm{1}
 \{V(\X{l})&=V(s)\}]
        = \quad \quad \quad\\
        \mathbb{E}_{\PIMB}[\frac{1}{p(X^{(l)})}\mathbbm{1}\{V(\X{l})&=V(s)\}]
    \end{split}
\end{equation*}
\end{lemma}
Lemma~\ref{lemma:equation_expectation} can be proved directly using the
definition. 
It is trivial to verify that the function $h^k_i(X^{(l)})$ in
Equation~\eqref{eq:indicator_function} is the linear combination of function
$\mathbbm{1}\{V(\X{l})=V(s)\}$. Using the linearity of expectation and the
result in Lemma~\ref{lemma:equation_expectation}, we have
\begin{equation*}
    \footnotesize{
\!\mathbb{E}_{\PIMB}\left[\frac{h^k_i(X^{(l)})}{p(X^{(l)})}\right]
= 
\mathbb{E}_{\PIMB}
    \left[\frac{h^k_i(X^{(l)})}{\alpha^k_i\pi_e(X^{(l)})}\right].
}
\end{equation*}
Hence, we can rewrite the estimator in Equation~\eqref{eq:graphletcounter} as
\begin{equation}\label{eq:css_unbiased_estimator}
    \frac{1}{n}\sum_{s=1}^n
    \frac{h^k_i(\X{l}_s)}{p(X^{(l)}_s)} \rightarrow C^k_i\ a.s.
\end{equation}
Similarly, we estimate graphlet concentration as 
\begin{equation}\label{eq:css_unbiased_concentration}
    \frac{\sum_{s=1}^{n}h^k_i(X^{(l)}_s)/p(X^{(l)}_s)}{\sum_{s=1}^{n}h^k(X^{(l)}_s)/p(X^{(l)}_s)}\
    \rightarrow c^k_i \ a.s.
\end{equation}

\noindent{\bf \em Remarks:}
The pseudo code of computing  $p(X^{(l)})$ is presented in
Appendix~\ref{app:computation_of_sampling_probability}.
The estimator in Equation~\eqref{eq:css_unbiased_estimator} corrects the bias
of graphlet samples using the sampling probability $p(X^{(l)})$ instead of
$\alpha^k_i\pi_e(X^{(l)})$.  It indicates that the probability that a subgraph
$s$ is generated by the random walk actually equals to $\sum_{\X{l}_j\in
\mathcal{C}(s)}\pi_e(\X{l}_j)$.

\noindent{\bf \em Example:} 
Table~\ref{table:new_coefficient}
lists the corresponding $p(X^{(l)})$ when we choose $SRW(1)$
for $3$-node graphlets and $SRW(2)$ for $4$-node graphlets.
Labels for nodes and edges are defined in Figure~\ref{fig:graphletexample}.
Note that an edge $e_{uv}$ in the graph $G$ is a node in $G^{(2)}$. The degree
of $e_{uv}$ in $G^{(2)}$ should be $d_u + d_v - 2$, i.e., $d_{e_{uv}} = d_u +
d_v - 2$. Here $d_u$ and $d_v$ are degrees of nodes $u$ and $v$ in $G$.  In
Table~\ref{table:new_coefficient}, the first column is the graphlet type of
subgraph induced by nodes in $X^{(l)}$.  The second column is the random walk
type. The third column is the state corresponding coefficient and the fourth
column is the sampling probability $p(\X{l})$. To further understand the
sampling probability, we give an example of triangle ($g^3_2$).  If we randomly
walk on $G$ and visit nodes $1, 2, 3$ sequentially, then the state we are
visiting is $\X{3}=(1, 2, 3)$.  Assume $\{1, 2, 3\}$ induces a triangle $\Delta$. 
We know that the corresponding states of $\Delta$
are $(1, 2, 3)$, $(3, 2, 1)$, $(1, 3, 2)$, $(2, 3, 1)$, $(2, 1,
3)$ and $(3, 1, 2)$. The sampling probability $p(X^{(l)})$ for $\X{3}=(1, 2,
3)$ is $\frac{1}{2|E|} (2/d_1 + 2/d_2 + 2/d_3)$ while
$\alpha^k_i\pi_e(X^{(l)})$ for $\X{3}$ is $\frac{6}{2|E|}\frac{1}{d_2}$.
Observe that $p(X^{(l)})$ {\em makes better use of the degree information for
the nodes in the subgraph}.
\begin{table}[t]
    \caption{Sampling probability $p(\X{l})$ for all $3, 4$-node
    graphlets.}\label{table:new_coefficient}
    \center
        \resizebox{0.45\textwidth}{!}{
    \begin{tabular}{|c|c|c|c|}
        \hline
        {\bf Graphlet} & ${\bf SRW(d)}$ & {\boldmath$\alpha_i^k / 2$} &
        {\boldmath $2|R^{(d)}|\cdot p(\X{l})/2$} \\ \hline \hline
        $g^3_1$ & \multirow{2}{*}{$SRW(1)$} & 1 & $1/d_{2}$\\
        \cline{1-1}\cline{3-4}
        $g^3_2$ & &3 &  $1 / d_1 + 1 / d_2 + 1 / d_3$\\ \hline
        $g^4_1$ & &1 &  $1 / d_{e_2}$\\ \cline{1-1}\cline{3-4}
        $g^4_2$ & \multirow{6}{*}{$SRW(2)$} & 3 &$\sum_{j=1}^3 1 /
        d_{e_j}$\\
        \cline{1-1}\cline{3-4}
        $g^4_3$ & &4 &  $\sum_{j=1}^4 1 / d_{e_j}$\\ 
        \cline{1-1}\cline{3-4}
        $g^4_4$ & &5 & $2 / d_{e_2} + 2 / d_{e_3} + 1 /
        d_{e_4}$\\ \cline{1-1}\cline{3-4}
        $g^4_5$ & &12 & $2 \sum_{j=1}^5 1 / d_{e_j} + 2 / d_{e_5}$\\ 
        \cline{1-1}\cline{3-4}
        $g^4_6$ & &24 & $4\sum_{j=1}^{6}1 / d_{e_j}$\\ \hline
    \end{tabular}
}
\end{table}

\noindent{\bf Bound on sample size.}
Define $W^{\prime}\triangleq \max_{X^{(l)}}1/p(X^{(l)})$.  When the sample size $n \geq \xi
(\frac{W^{\prime}}{C^k_i})\frac{\tau}{\epsilon^2}(\log\frac{\norm{\varphi}_{\PIMB}}{\delta})$,
the estimate in Equation~\eqref{eq:css_unbiased_concentration} is within $(1\pm
\epsilon)c^k_i$ with probability at least $1-\delta$.  Here, $\xi$ is a constant
independent of $\epsilon, \delta$. $\tau$ is the mixing time $\tau(1/8)$ of the
random walk. $\varphi$ is the initial distribution.  Since
$\max 1/p(X^{(l)})\leq \max 1/\alpha^k_i\pi_e(\X{l})$, the bound on the sample
size for the new estimator in Equation~\eqref{eq:css_unbiased_concentration} is
smaller.

\noindent{\bf Efficiency.} For a fixed sample budget, the efficiency of the
sampling scheme is determined by the asymptotic variance. 
Since the covariance between random variables is difficult to analyze, we
analyze the variance of independent random variables.
Lemma~\ref{lemma:css_variance} indicates that {\em if all states are
independent, CSS sampling scheme is more efficient than the basic one in
Section~\ref{sec:framework}.}
\begin{restatable}[]{lemma}{cssvariance}\label{lemma:css_variance}
    The variance of function $h^k_i(\X{l})/p(X^{(l)})$ is smaller than that of
    function $h^k_i(\X{l})/\left(\alpha^k_i\pi_e(X^{(l)})\right)$ under the
    same stationary distribution $\PIMB$. Specifically, we have 
    \begin{equation*}
        \text{Var}_{\PIMB}\left[h^k_i(\X{l})/p(X^{(l)})\right]\leq
        \text{Var}_{\PIMB}\left[h^k_i(\X{l})/(\alpha^k_i\pi_e(X^{(l)}))\right].
    \end{equation*}
\end{restatable}
See Appendix~\ref{app:css_variance} for more details on the proof.

\subsection{Non-backtracking Random Walk}
Many techniques have been proposed to improve the efficiency of random
walk-based algorithms, for example, non-backtracking random
walk~\cite{lee2012beyond}, random walk leveraging walk
history~\cite{DBLP:journals/pvldb/Zhou0D15}, rejection controlled
Metropolis-Hasting random walk~\cite{lirandom}, random walk with
jump~\cite{Xu2014}, etc. In this subsection, we introduce non-backtracking
random walk (NB-SRW) to our estimation framework {\em as an example to show how
to integrate these techniques with our framework}.

The basic idea of NB-SRW is to avoid backtracking to the previously visited node.
Due to the dependency on previously visited node, the random walk on $G^{(d)}=
(H^{(d)}, R^{(d)})$ is not a Markov chain on state space $H^{(d)}$. However, we
can define an augmented state space $\Omega=\{(i, j): i, j\in H^{(d)}, {\ s.t.\
} (i, j)\in R^{(d)}\}\subseteq H^{(d)}\times H^{(d)}$. The transition matrix
$\mathbf{P}^\prime\triangleq\{P^\prime(e_{ij}, e_{lk})\}_{e_{ij}, e_{lk}\in
\Omega}$ for the NB-SRW is defined as follows
\begin{equation*}
    \mathrm{P}^{\prime}(e_{ij}, e_{jk})=\left\{
         \begin{array}{lr}
             \frac{1}{d_j-1}, & \text{if } i\neq k \text{ and } d(j)
             \geq 2, \\
             0, &\text{if } i=k\text{ and } d(j)\geq 2, \\
             1, &\text{if } i=k \text{ and } d(j) = 1.
     \end{array}\right.
\end{equation*}
All other elements of matrix $\mathbf{P}$ are zeros. Let $\pmb{\pi}^\prime$ be
the stationary distribution of the NB-SRW. A useful fact is that NB-SRW
preserves the stationary distribution of the original random walk, i.e.,
$\pi^\prime(i)=d_i/2|R^{(d)}|$ and $\pi^\prime(e_{ij})=1/2|R^{(d)}|$. To apply
NB-SRW, we just need to replace our previously used simple random walk with
NB-SRW. The estimator in Equation~\eqref{eq:concentration_estimator}
and~\eqref{eq:css_unbiased_concentration} can still be used except that we need
to replace the $\pi_e(\X{l})$ with $\pi^{\prime}_e(X^{(l)})$. 
Define the nominal degree for $X_i\in H^{(d)}$ as $d^\prime_{X_i}=\max\{d_{X_i}
- 1, 1\}$. For any $X^{(l)} = (X_1, \cdots, X_l)$, the $\pi^{\prime}_e(\X{l})$ can be
computed by substituting $d_X$ with $d^{\prime}_X$ in
Equation~\eqref{eq:stationary_distribution}.

Applying NB-SRW helps us eliminate some ``{\em invalid}'' states from the state
space.  For example, if we want to estimate $3$-node graphlet concentration
using $SRW(1)$, we need to walk for $3$ steps on $G$ to collect a sample. 
It is possible for us to get only $2$ distinct nodes from
$3$ steps. We call such samples as invalid samples.
Figure~\ref{fig:validsample_example} shows an example of valid sample
and invalid sample. The invalid samples do not contribute to the estimation.
If we apply NB-SRW here, it is less likely to get such invalid samples. Hence
NB-SRW helps to improve the estimation efficiency of our framework.
\begin{figure}[htb]
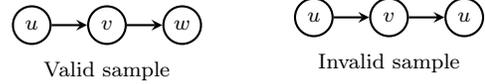

    \centering
    \scalebox{1} { \NBSRWEXAMPLE }
    \caption{Example of ``valid'' \& ``invalide'' 
    samples}\label{fig:validsample_example}
\end{figure}

\section{Implementation Details}\label{sec:implementation}

\noindent{\bf Populate Neighbors of Graphlet.}
We explain how to obtain neighbors of currently visited state (subgraph) $s\in
H^{(d)}$ on the fly.  Obtaining a uniformly chosen neighbor of a node in $G$ or
$G^{(2)}$ takes $O(1)$ time. We give details about the random walk on $G^{(2)}$.
The set of neighbors of $e_{uv}$ is $N(e_{uv})=\{e_{uw}: w\in N(u)\backslash
v\}\cup \{e_{vz}: z\in N(v)\backslash u\}$.  Recall that $N(v)$ denotes the set
of neighbors of $v\in V$.  To ensure each neighbor of $e_{uv}$ is chosen
uniformly, we first select one of $u$ and $v$ with probability $d_u / (d_u +
d_v)$ and $d_v / (d_u + d_v)$ respectively.  Suppose we have chosen node $u$,
we then randomly select a node $w\in N(u)$. If $w\neq v$, $e_{uw}$ is
proposed as the next step after $e_{uv}$. Otherwise, we restart the process until we
obtain a neighbor of $e_{uv}$.  Based on above discussion, we know that getting
a uniformly sampled neighbor of a state in $G^{(2)}$ can be done in constant
time.

To obtain a randomly chosen neighbor of $s$ in $G^{(d)}$, we can replace one
node $v_i$ in $V(s)$ with a node $v_j\in \cup_{v\in V(s)\backslash v_i} N(v)$
and ensure the connectivity of this new subgraph induced by node set
$\{v_j\}\cup V(s)\backslash v_i$. Here $V(s)$ is the node set of the subgraph
$s$. However, to ensure the neighbors of $s$ are uniformly sampled, we need to
generate all neighbors of $s$ when $d >2$, which requires $d$ merge operations
over $d - 1$ adjacent lists of nodes in the currently visited subgraph. Hence,
the time complexity of selecting a random neighbor for a state in $G^{(d)}$ is
simply $O(d^2\frac{|E|}{|V|})$ when $d > 2$.

\noindent{\bf Identify Graphlet Types.}
According to Algorithm~\ref{algo:unbiased_estimate}, we also need to identify
the graphlet type of obtained samples at each step.
Assume the last $l + 1$ steps of the random walk on $G^{(d)}$ are $\{X_0, X_1,
\cdots, X_l\}$. Our task is to identify the graphlet type of the subgraph
induced by $V_k = \cup_{1\leq i\leq l}V(X_i)$.  
Let $G_k = (V_k, E_k)$ denote the
subgraph induced by $V_k$.  
The first step of identifying
the graphlet type of $G_k$ is to obtain the $E_k$.  
Assume the adjacent lists are stored in sorted
order. The naive solution to obtain $E_k$ needs $k(k-1)/2$ binary searches.  
Now, we show that $k - 1$ binary search operations are sufficient to determine
$E_k$.  The core idea is to reuse previous computation results.  
Let $V^{\prime}_k =
\cup_{0\leq i\leq l-1}V(X_i)$ and $G^{\prime}_k = (V^{\prime}_k, E^{\prime}_k)$
be the subgraph induced by $V^{\prime}_k$.  
Suppose we already obtain
$E^{\prime}_k$.  According to the definition of $SRW(d)$, the
set $V_k = \cup_{1\leq i\leq l}V(X_i)$ can be rewritten as
$V^{\prime}_k\backslash V_\text{out}\cup V_\text{in}$ where $V_\text{out}=
V(X_0)\backslash V(X_1)$ and $V_\text{in} = V(X_l)\backslash V(X_{l-1})$. Here,
$|V_\text{out}| = 1$ and $|V_\text{in}| = 1$. The key observation is that we
only need to determine the adjacent relationship between $v_\text{in}\in
V_\text{in}$ and $v\in V_k\backslash V_\text{out}$, which requires $k-1$ binary
searches on $N(v_\text{in})$ since the adjacent relationship between
nodes in $V_k\backslash V_\text{in}$ can be obtained directly from
$E^{\prime}_k$. 
The time complexity of obtaining the edge lists of graphlet
samples at each step is $O(k\log \frac{|E|}{|V|})$.

With the set $E_k$, we can now determine the graphlet type of the subgraph $G_k =
(V_k, E_k)$. Note that the time to determine the graphlet types of samples
grows exponentially with $k$ since the number of distinct graphlets grows
exponentially with $k$.  However, for $k=3, 4, 5$-node graphlets, we can
determine the graphlet types more easily with the \textit{degree-signature} of
different graphlets~\cite{bhuiyan2012guise}, i.e., for the subgraph $G_k =
(V_k, E_k)$, we just need to compute the degree of each node and derive the
degree-signature, then compare it with the degree-signature of different
graphlet types~\cite{bhuiyan2012guise}. 

Note that our framework does not need any preprocessing of the graph data.
The time complexity of our
framework is $O(n)$ when $d\leq 2$ and $O(n d^2 \frac{|E|}{|V|})$ when $d >
2$ for the $k$-node graphlets. Here $n$ is
the random walk steps.

\section{Experimental Evaluation}\label{sec:experiment}
We evaluate the performance of our framework on $3, 4, 5$-node
graphlets. The algorithms are implemented in \verb!C++! and we run experiments
on a Linux machine with Intel 3.70GHz CPU. We aim to answer the following
questions.
\begin{itemize}
    \item[Q1:] How accurate is our framework? Do the optimization techniques
        really improve the accuracy?
    \item[Q2:] How does the random walk on $G^{(d)}$ affect the performance? What
        is the best parameter $d$ for $3, 4, 5$-node graphlets?
    \item[Q3:] Does our framework provide more accurate estimation than the
        state-of-the-art methods?
\end{itemize}

\subsection{Experiment Setup}
We use publicly available real-world networks to evaluate our framework. We
focus on undirected graphs by removing the directions of edges if the graphs
are directed. We only retain the largest connected component (LCC) of the
graphs and discard the rest nodes and edges. The detailed information about the
LCC of the graphs is reported in Table~\ref{table:dataset}.

The exact graphlet concentration is obtained through well-tuned
enumeration methods~\cite{ahmed2015icdm, hovcevar2014combinatorial}.  For
$5$-node graphlets, the ground-truth value is only computed for the four
smaller datasets due to the extremely high computation cost. As an example, we
state the exact concentration of the $3, 4, 5$-node cliques (i.e., $c^3_2$,
$c^4_6$, and $c^5_{21}$) in Table~\ref{table:dataset}, and we can see that the
$3, 4, 5$-node cliques take a relatively low percentage.

\begin{table}[t]
\centering
\caption{Datasets} \label{table:dataset}
\resizebox{0.48\textwidth}{!}{
\begin{tabular}{|l|c|c|c|c|c|c|c||}
    \hline
    Graph     & $|V|$   & $|E|$ & \specialcell{$c^3_2$
    \\ ($10^{-2}$)} & \specialcell{$c^4_6$ \\ ($10^{-3}$)} &
    \specialcell{$c^5_{21}$\\ ($10^{-5}$)} \\ \hline \hline
    BrightKite~\cite{konect} 
    & 57K & 213K & 3.98 & 1.447   & 4.661 \\
    Epinion~\cite{snapnets}    
    & 76K & 406K & 2.29 & 0.225   & 0.147 \\
    Slashdot~\cite{snapnets}   
    & 77K & 469K & 0.82 & 0.092   & 0.115 \\
    Facebook~\cite{konect}  
    & 63K & 817K & 5.46 & 1.419   & 2.511 \\
    Gowalla~\cite{snapnets}   
    & 197K& 950K & 0.80 & 0.008   & - \\
    Wikipedia~\cite{konect}  
    & 1.9M& 36.5M& 0.10 & 0.00009 & - \\
    Pokec~\cite{konect}
    & 1.6M& 22.3M &1.6 & 0.035 & - \\
    Flickr~\cite{konect}     
    & 2.2M& 22.7M& 3.87 & 0.886   & - \\
    Twitter~\cite{nr} 
    & 21.3M &265M & 0.86& 0.0166  &-\\ 
    Sinaweibo~\cite{nr} 
    & 58.7M &261M & 0.03& 0.00008&-\\
    \hline
    \end{tabular}
}
\end{table}

We use the following normalized root mean
square error (NRMSE) to measure the estimation accuracy: 
\begin{equation*} 
    {\small
    \text{NRMSE}(\hat{c}^k_i)
    \triangleq\frac{\sqrt{\mathbb{E}[(\hat{c}^k_i - c^k_i)^2]}}{c^k_i} =
    \frac{\sqrt{\text{Var}[\hat{c}^k_i] + (c^k_i - \mathbb{E}[\hat{c}^k_i])^2}}{c^k_i},
    }
\end{equation*} 
where $\hat{c}^k_i$ is the estimated value and $c^k_i$ is the ground-truth.
NRMSE$(c^k_i)$ is a combination of variance and bias of the estimate
$\hat{c}^k_i$, both of which are important to characterize the accuracy of
the estimator.

The names of the methods are given in the following way.  \textsf{SRWd}
represents random walks on $G^{(d)}$. If the method also integrates the
optimization techniques \textit{corresponding state sampling} (\textsf{CSS}) and
\textit{non-backtracking random walk} (\textsf{NB-SRW}), we append \textsf{CSS} and
\textsf{NB} at the end of the method name, respectively.  For example,
\textsf{SRW1CSSNB} means that we perform the random walk on $G^{(1)}$ (i.e., $G$) and
use both techniques \textit{CSS} and \textit{NB-SRW} for further optimization.

\subsection{Framework Evaluation}
\subsubsection{Accuracy}
\noindent{\bf Comparison between different random walks.} We first demonstrate
the effects of the parameter $d$ and the optimization techniques on the estimation
accuracy.  The estimation results are presented in
Figure~\ref{fig:different_random_walks} for all the graphs whose ground truth
has been obtained. Note that only graphlets $g^3_2$, $g^4_6$, $g^5_{21}$ are
presented since they have the smallest concentration value among $3, 4, 5$-node
graphlets respectively and are observed to have the least accurate estimates.
The NRMSE is estimated over 1,000 independent simulations except that the NRMSE
of \textsf{SRW4} is only estimated over 100 simulations since the random walk
on $G^{(4)}$ is relatively slow.  We do not study \textsf{SRW1} for $4, 5$-node
graphlets because $\alpha^4_2, \alpha^5_2, \alpha^5_3, \alpha^5_6$ are zeros
with \textsf{SRW1}.  The sample size, i.e., the random walk steps, equals to
$20$K for all methods in the framework. We summarize our findings as follows.
\begin{itemize}[leftmargin=*]
    \item The method \textsf{SRW1CSSNB}, i.e., random walk on $G$ with
        optimization techniques \textsf{CSS} and \textsf{NB-SRW}, has the highest
        accuracy in estimating the concentration of $3$-node graphlets. The
        method \textsf{SRW2CSS} has the best performance in estimating $4,
        5$-node graphlet concentration.
    \item The best methods in our framework provide accurate estimates.  The
        NRMSE of \textsf{SRW1CSSNB} for graphlet $g^3_2$ is in the
        range $0.025\sim 0.13$. The NRMSE of \textsf{SRW2CSS} for
        graphlets $g^4_{6}, g^5_{21}$ is in the range $0.08\sim 4.3$, and $0.20
        \sim 0.86$ respectively. Note that we only use 20K random walk steps.
        The sample size is small compared with the graph size. 
        E.g., we only exploit 0.03\% nodes of Sinaweibo.
    \item For the same graphlets, the random walk on $G^{(d)}$ with smaller $d$
        outputs better estimates. E.g., \textsf{SRW1} outputs estimates of
        $c^3_2$ which has $3.8\times$ smaller NRMSE than that of \textsf{SRW2}
        for Twitter; \textsf{SRW2} produces estimates of $c^4_6$ with
        $10\times$ smaller NRMSE than \textsf{SRW3} for Gowalla. In conclusion,
        we should choose $d = \{1, 2, 2\}$ for $3, 4, 5$-node graphlets respectively.
    \item The optimization technique \textsf{CSS} improves the accuracy of estimates 
        a lot while the performance gain of \textsf{NB-SRW} is
        negligible. For example, \textsf{SRW1CSS} reduces the NRMSE of
        \textsf{SRW1} more than $3$ times for Wikipedia and
        Sinaweibo when estimating the triangle ($g^3_2$) concentration.  
\end{itemize}

\begin{figure*}[t]
    \centering
    \setcounter{subfigure}{0}
    \subfloat[][Triangle ($g^3_2$, \scalebox{0.4}{\GTHREE{2}})]
    {\label{fig:3nodegraphletnrmse}
        \includegraphics[width=0.96\textwidth, height=0.12\textwidth]{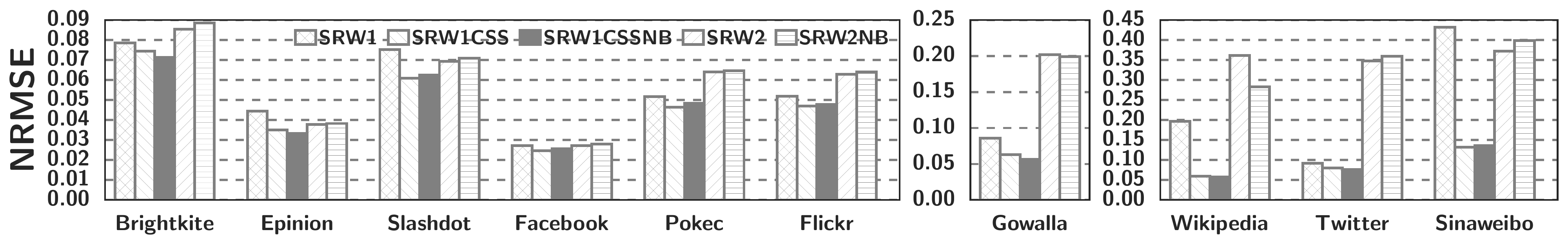}}\\
    \subfloat[][$4$-node clique ($g^4_6$, \scalebox{0.4}{\Gfour{6}})]
    {\label{fig:4nodeid6graphletnrmse}
        \includegraphics[width=0.65\textwidth, height=0.12\textwidth]{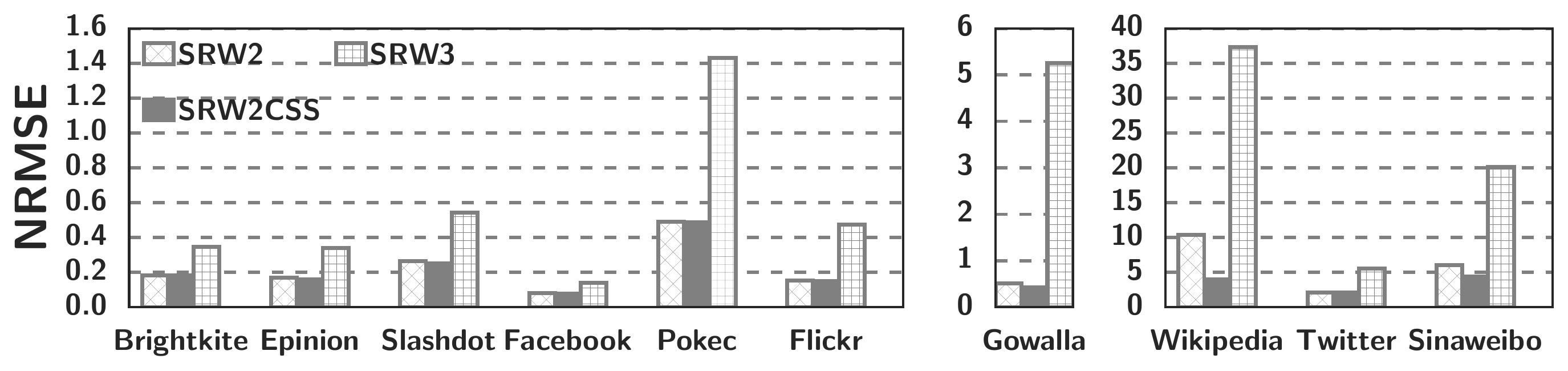}}
    \subfloat[][5-node clique ($g^5_{21}$,
    \adjustbox{valign=m}{\scalebox{0.6}{\GFIVE{21}}})]
    {\label{fig:5id21graphletnrmse}
    \includegraphics[width=0.31\textwidth,height=0.12\textwidth]{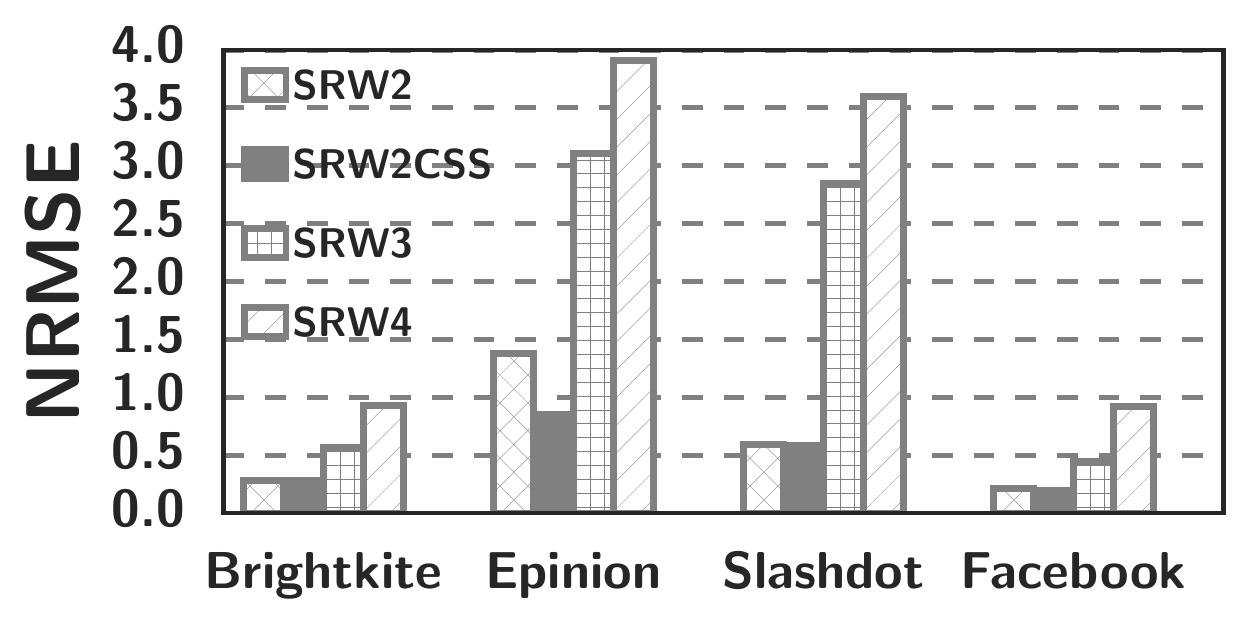}}
    \caption{NRMSE of concentration estimates (sample size =
    20K).}\label{fig:different_random_walks}
\end{figure*}

\begin{figure}[htb]
    \centering
    \captionsetup[subfigure]{justification=centering, margin={0.0cm, 0.0cm}}
    \subfloat[][Weighted concentration]{
    \includegraphics[width=0.23\textwidth,
    height=0.15\textwidth]{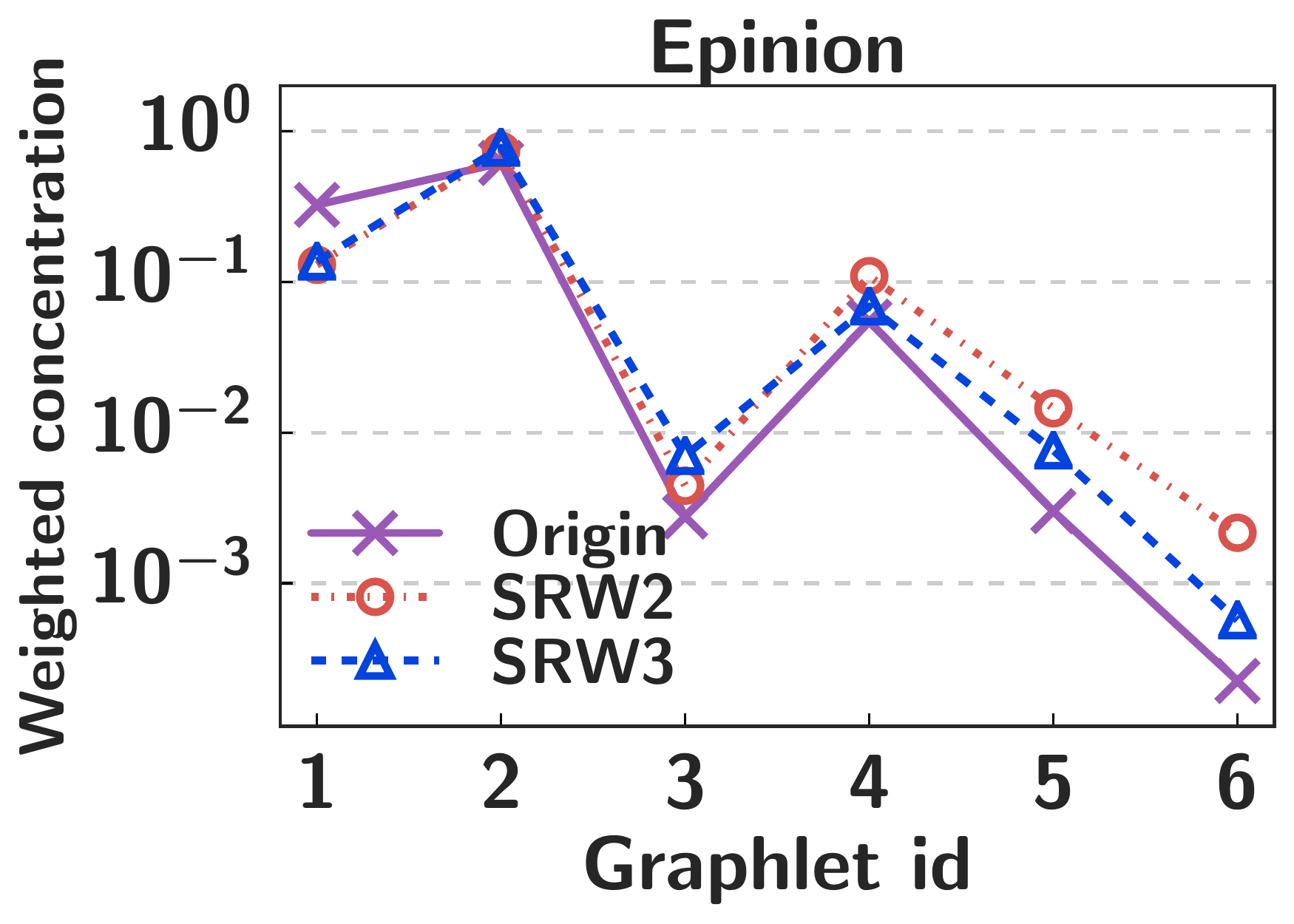}\label{subfig:weightedconcentration}
    }
    \subfloat[][NRMSE] {
    \includegraphics[width=0.23\textwidth,
    height=0.15\textwidth]{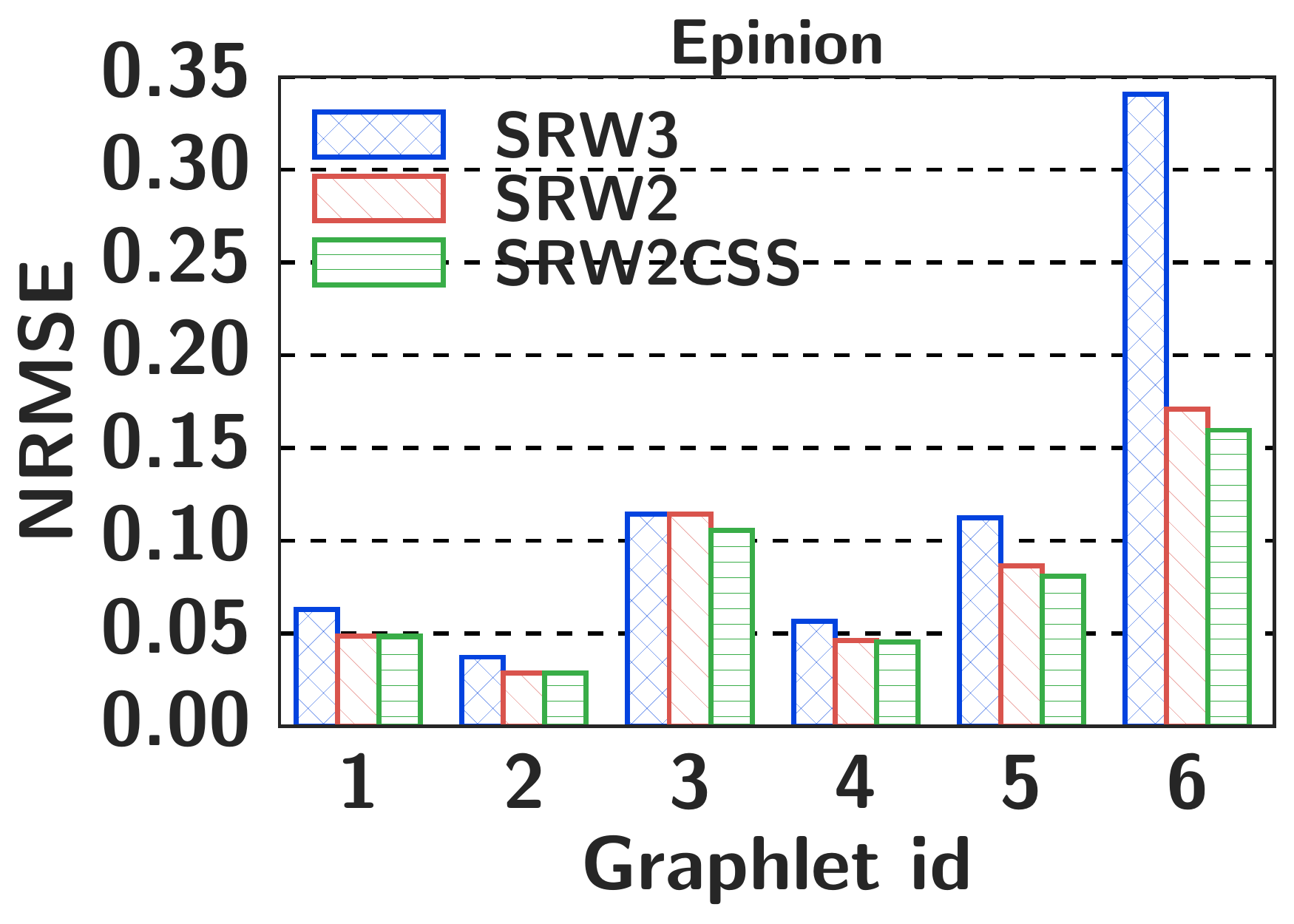}\label{subfig:individualnrmse}
    }\\
    \caption{Relationship between the weighted concentration and the accuracy
    (sample size = 20K).}\label{fig:4nodegraphletsimulation}
\end{figure}

\noindent\textbf{Weighted concentration and accuracy.} We now introduce the
concept of {\it weighted concentration} to explain how the parameter $d$
affects the performance.
From Equation~\eqref{eq:graphletcounter}, we know that $\frac{1}{n} \sum_{s =
1}^{n}\frac{h^k_i(X^{(l)}_s)}{\pi_e(X^{(l)}_s)}\rightarrow \alpha^k_i
C^k_i$. To further understand the performance of our framework, we define the
\textit{weighted concentration} for graphlet $g^k_i$ as $\alpha^k_i
C^k_i/(\sum_{j=1}^{m}\alpha^k_j C^k_j)$.
As an example, we plot the weighted concentration of $4$-node graphlets for 
Epinion in Figure~\ref{subfig:weightedconcentration}. From
Figure~\ref{fig:4nodegraphletsimulation}, we know the parameter $d$ and
the concentration value affect the performance of our framework in the following way.
\begin{itemize}[leftmargin=*]
    \item {\bf \em Effect of the parameter $d$.} Compared with the original
        concentration, the weighted version lifts the percentage of relatively
        rare graphlets, i.e., $g^4_3$
        (\adjustbox{valign=m}{\scalebox{0.4}{\Gfour{3}}}), $g^4_5$
        (\adjustbox{valign=m}{\scalebox{0.4}{\Gfour{5}}}), and $g^4_6$
        (\adjustbox{valign=m}{\scalebox{0.4}{\Gfour{6}}}).  For the graphs
        Epinion, the weighted concentration
        of \textsf{SRW2} is much larger than that of \textsf{SRW3} for
        graphlets $ g^4_5, g^4_6$, while for graphlet $g^4_3$, the weighted
        concentration of \textsf{SRW3} is slightly higher than that of
        \textsf{SRW2}. For example, the weighted concentration for $g^4_6$ with
        \textsf{SRW2} is about $8\times$ higher than the original one while
        \textsf{SRW3} only increases the concentration $2\times$ higher. In
        other words,  \textsf{SRW2} increases the probability of getting a
        sample of $g^4_6$ about 8$\times$ higher compared with uniform sampling
        of graphlets, while \textsf{SRW3} only increases the probability about
        $2\times$ higher. Consequently, the NRMSE of \textsf{SRW2} in
        estimating $c^4_6$ is $2\times$ smaller than that of \textsf{SRW3}.
        From Theorem~\ref{thm:union_bound} we know that more samples are needed
        to achieve specific accuracy for
        graphlets with smaller $\alpha^k_iC^k_i$. Hence the error of the
        estimation for rare graphlets is the major error source. If we are able
        to assign rare graphlets higher weighted
        concentration, the overall performance is less likely to degenerate.
        {\em Based on above discussion, we conclude that random walks on $G^{(d)}$
        with smaller $d$ have better overall performance since they have a higher
        chance of getting the relatively rare graphlets.}

    \item {\bf \em Effect of the concentration value.} From
        Figure~\ref{subfig:individualnrmse} we can see that \textsf{SRW2} and
        \textsf{SRW2CSS} perform better than \textsf{SRW3} for all the $4$-node
        graphlets except $g^4_3$ (because the weighted concentration of $g^4_3$
        computed with \textsf{SRW3} is higher than that of \textsf{SRW2}).
        Besides, the smaller the concentration value, the higher the estimation
        error, which is consistent with our analysis in
        Theorem~\ref{thm:union_bound}. 
\end{itemize}

\subsubsection{Convergence}
To show the convergence properties of the methods, we vary the
sample size from 2K to 20K in increment of 1K. We present the simulation
results in Figure~\ref{fig:convergence_analysis} for $3, 4, 5$-node cliques
since they have the smallest concentration value and the least accuracy.  Due
to the space limitation, we only choose 6 representative graphs in the datasets
for the presentation.  We summarize the observations as follows. 
\begin{itemize}[leftmargin=*]
    \item The estimates are more concentrated around the ground-truth as we
        increase the sample size. 
    \item \textsf{SRW1CSSNB} exhibits consistent best performance in estimating
        $c^3_2$. \textsf{SRW2CSS} has consistent best performance in estimating
        $4, 5$-node clique concentration.
    \item There are spikes in the line plot of NRMSE v.s. sample size. This is
        caused by burn-in period of the random walks and inadequate simulation
        times.
\end{itemize}
\begin{figure}[htb]
    \centering
    \captionsetup[subfigure]{justification=centering, margin={0.0cm, 0.0cm}}
    \setcounter{subfigure}{0}
    \subfloat[][Triangle (\scalebox{0.4}{\GTHREE{2}})]
    {\includegraphics[width=0.24\textwidth, height=0.15\textwidth]{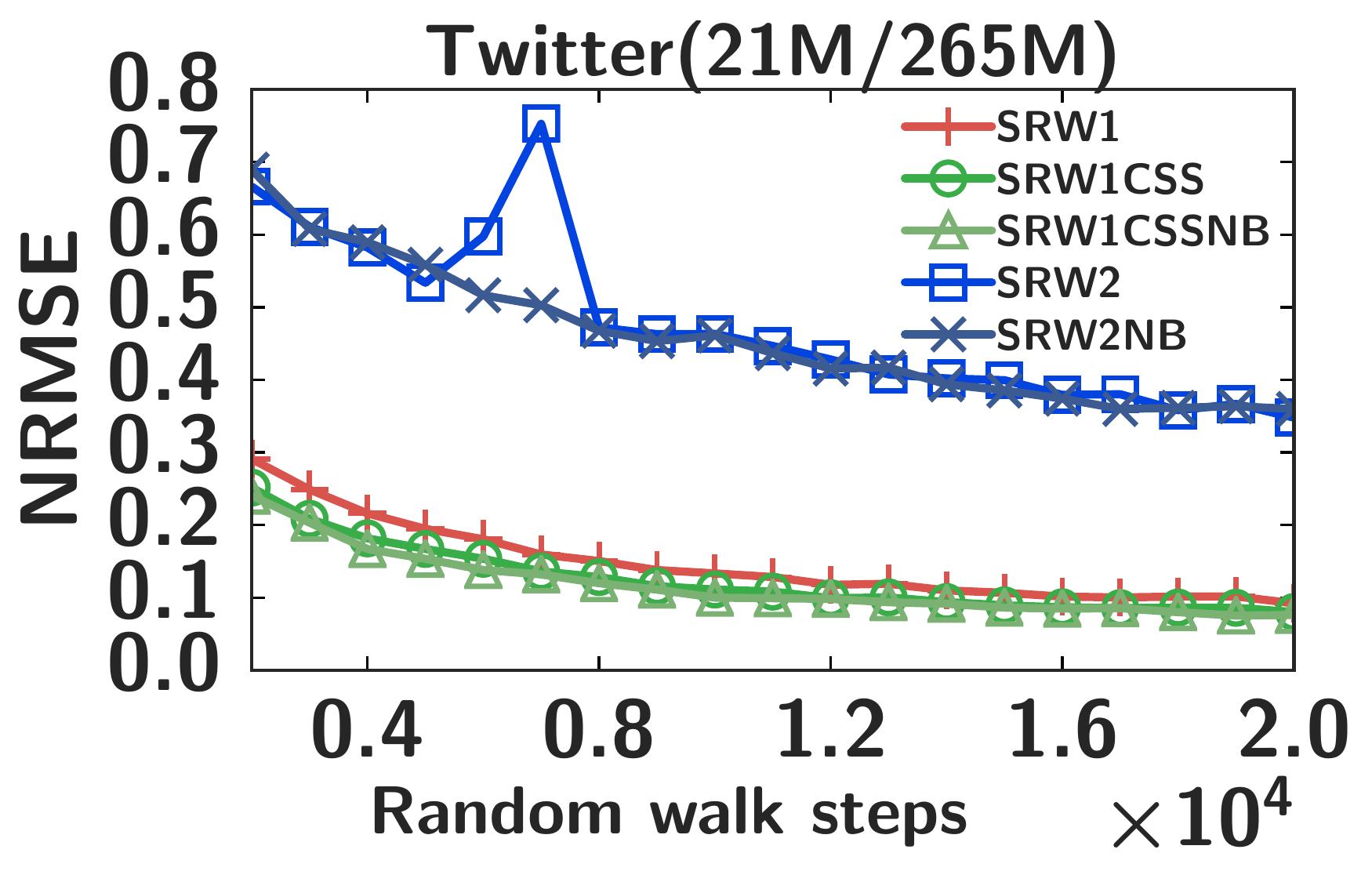}
    \includegraphics[width=0.24\textwidth, height=0.15\textwidth]{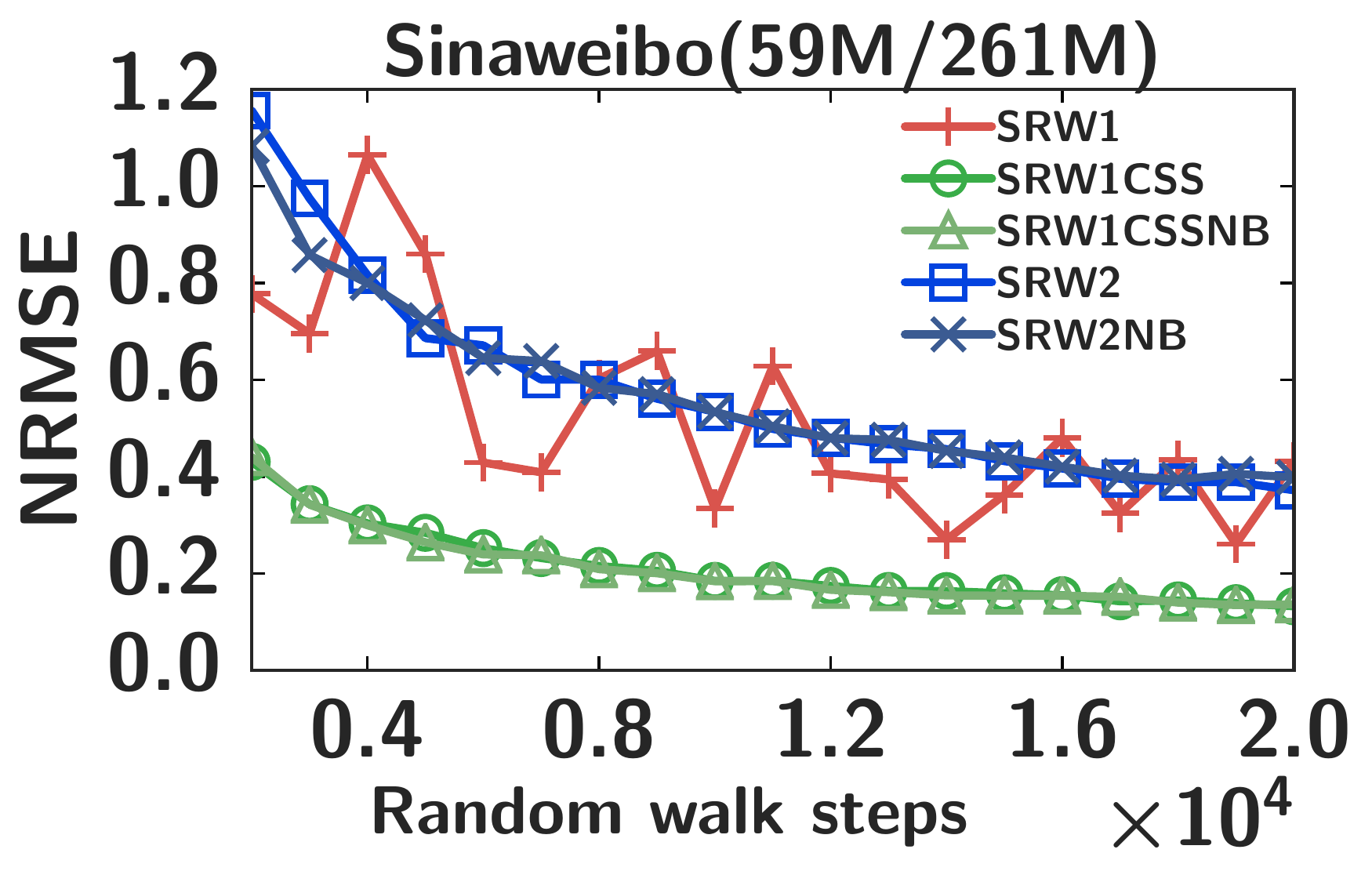}
    \label{subfig:triangle_converge}
    }\\
    \subfloat[][$4$-node clique (\scalebox{0.4}{\Gfour{6}})]
    {\includegraphics[width=0.24\textwidth, height=0.15\textwidth]{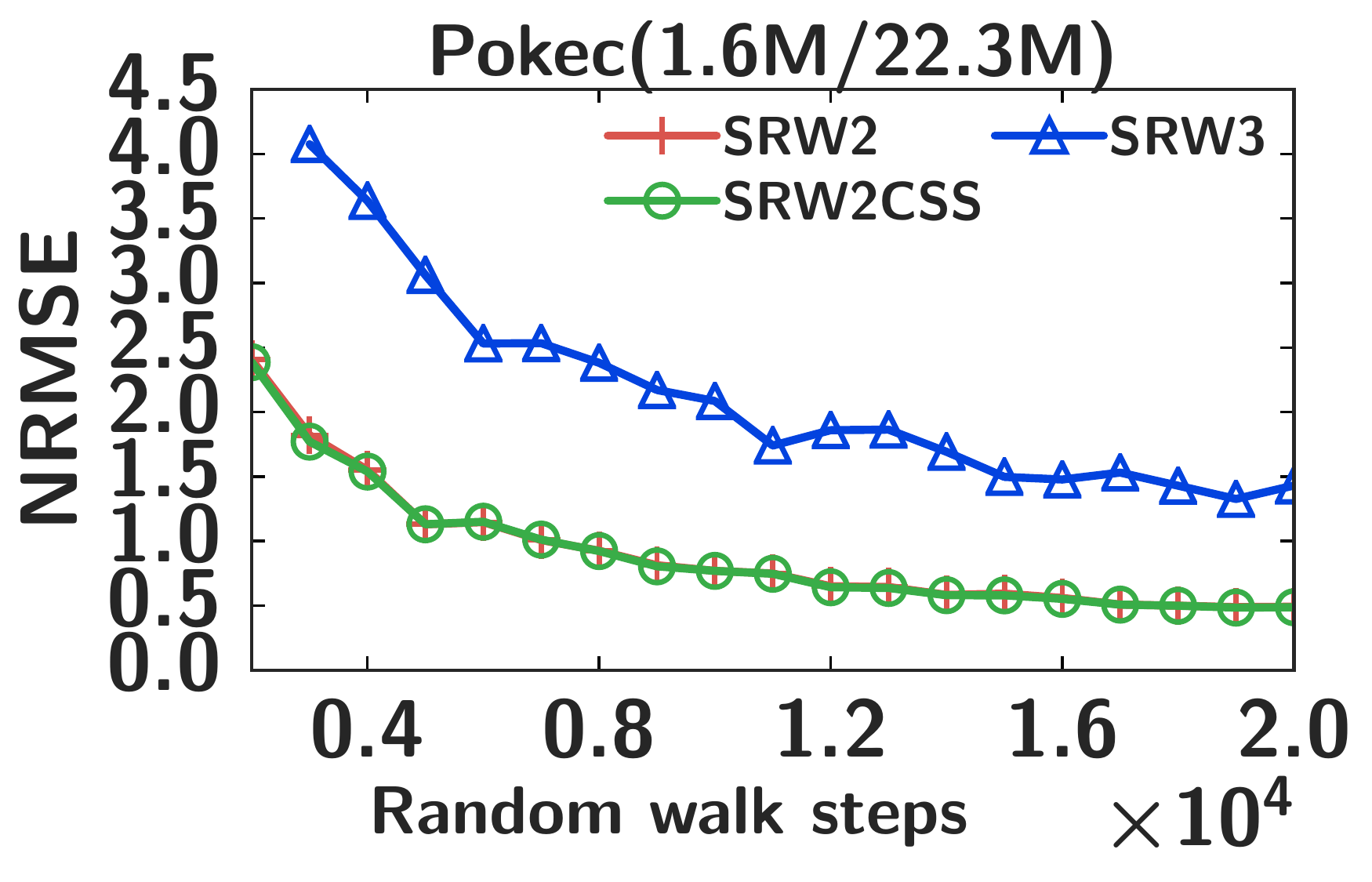}
    \includegraphics[width=0.24\textwidth, height=0.15\textwidth]{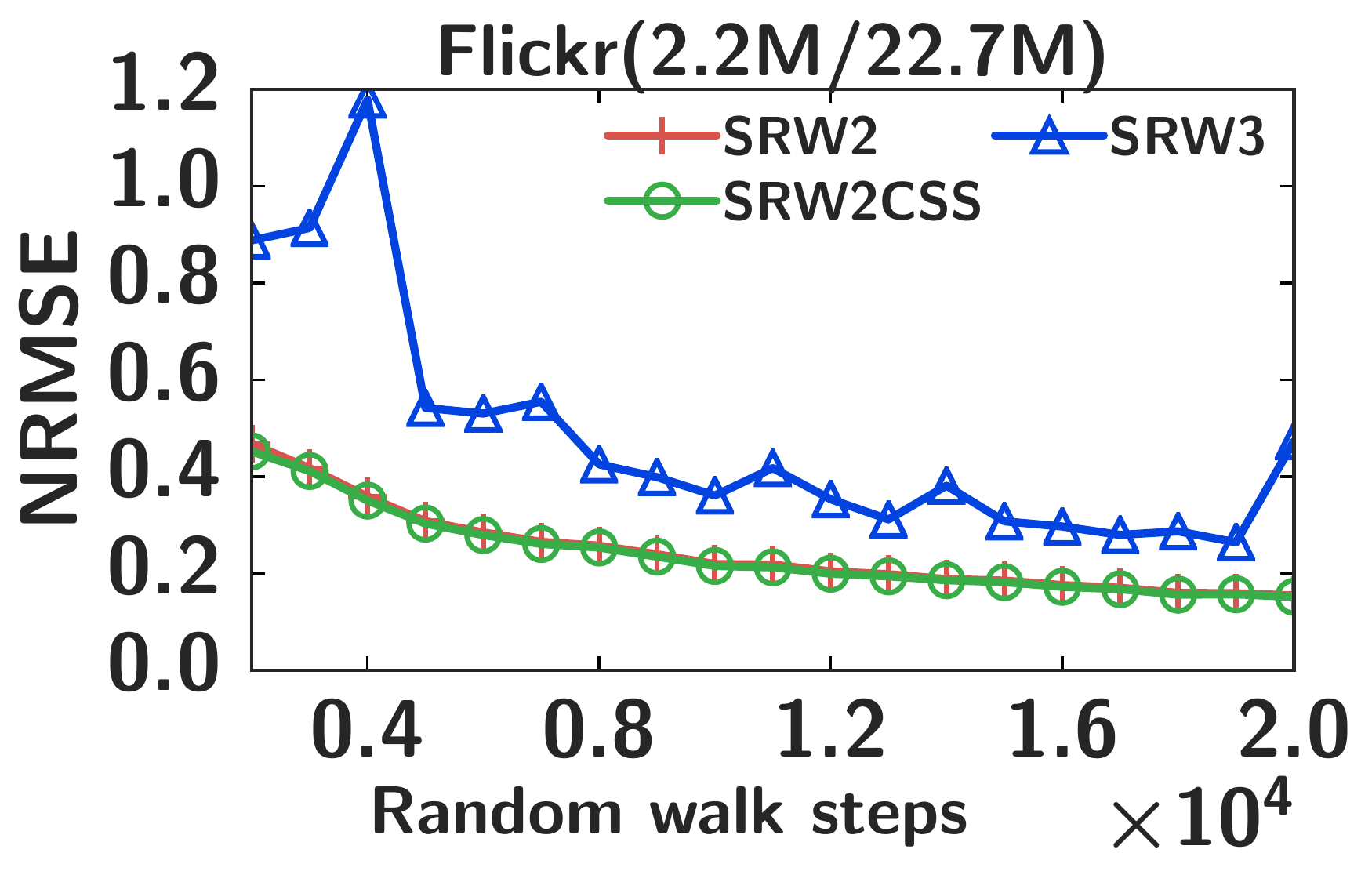}}\\
    \subfloat[][$5$-node clique (\adjustbox{valign=m}{\scalebox{0.6}{\GFIVE{21}}})]
    {\includegraphics[width=0.24\textwidth, height=0.15\textwidth]{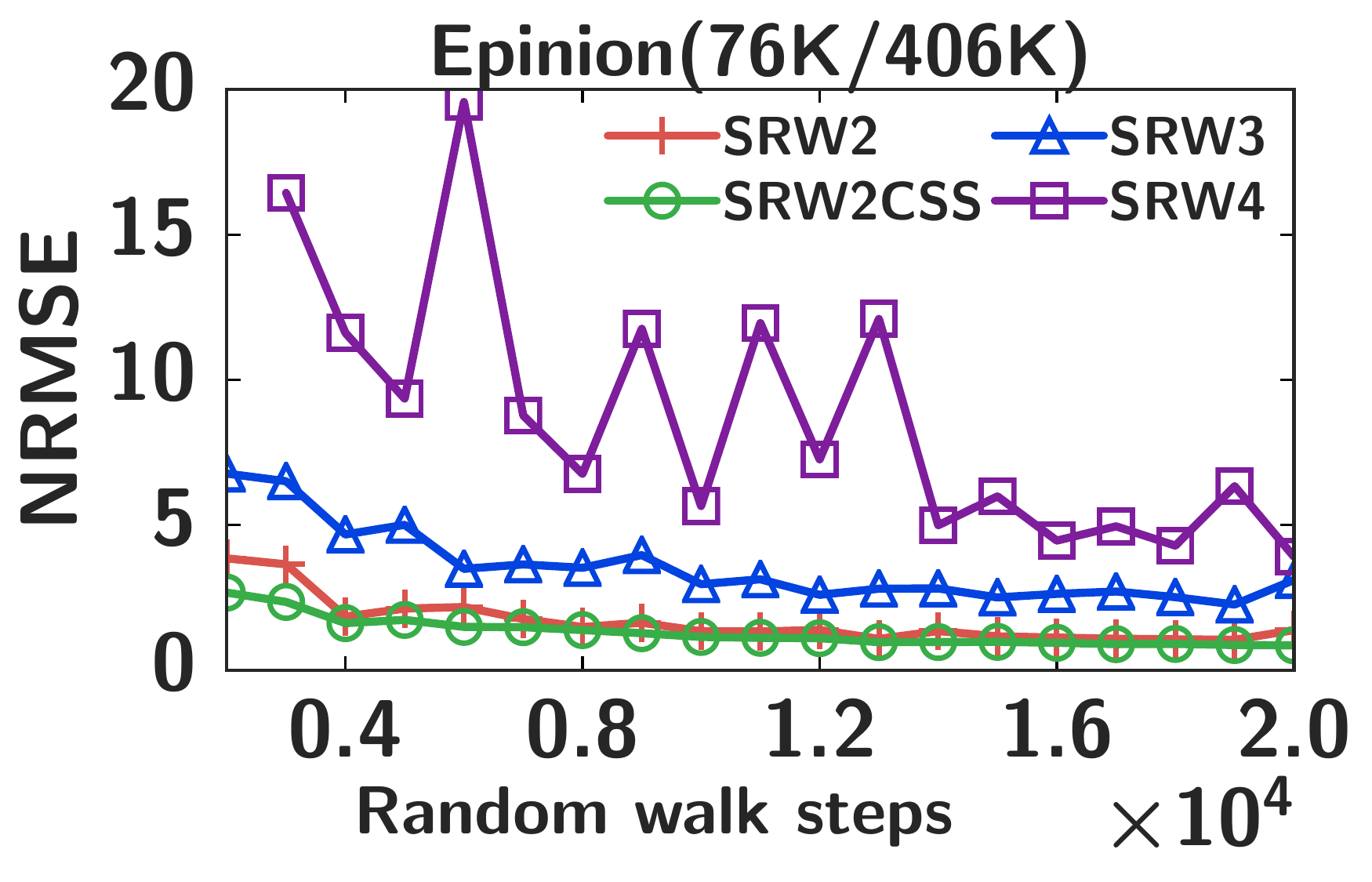}
    \includegraphics[width=0.24\textwidth, height=0.15\textwidth]{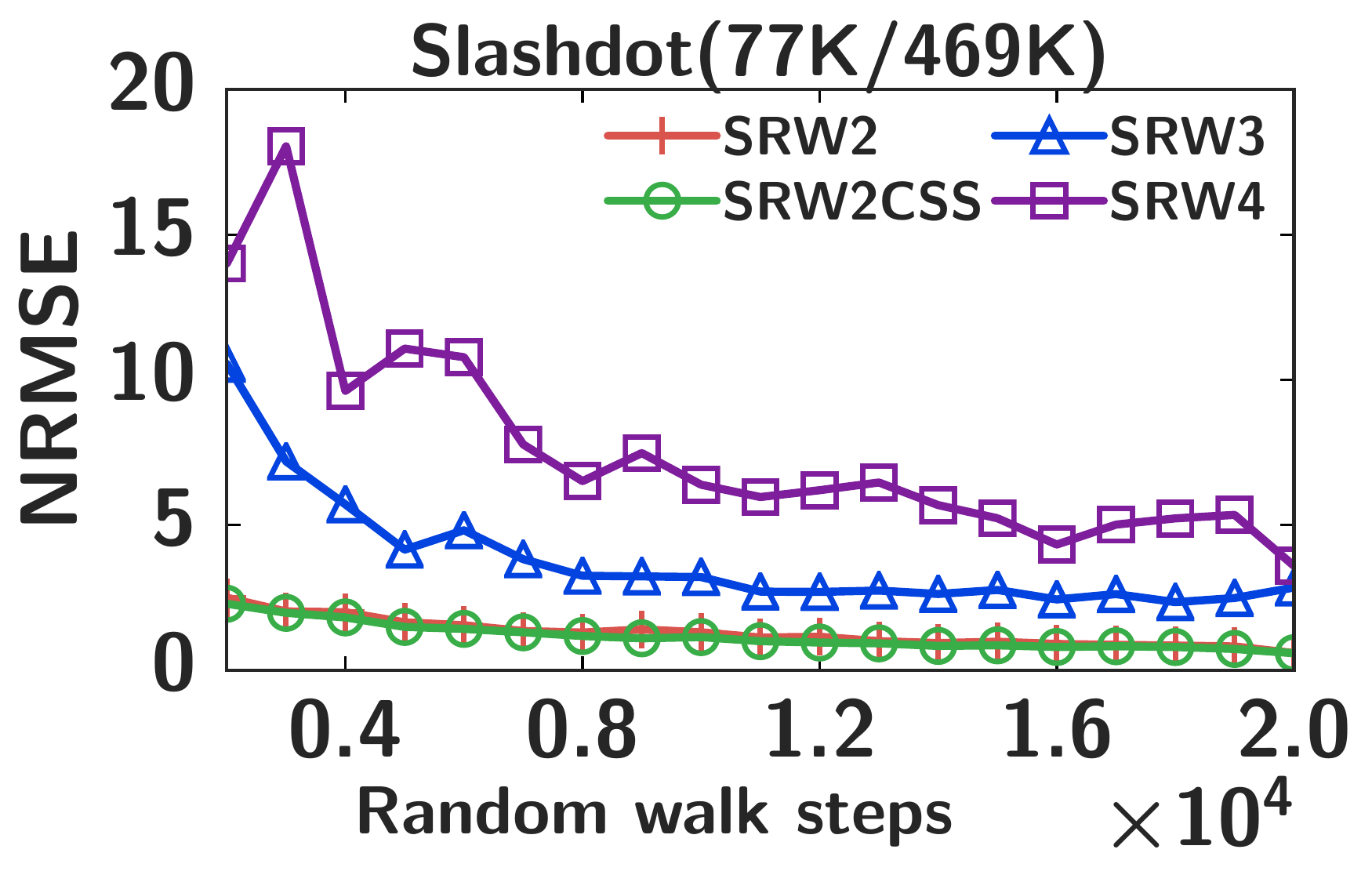}}
    \caption{Convergence of the estimates.}\label{fig:convergence_analysis}
\end{figure}

\subsubsection{Running time}
In Table~\ref{table:timecost}, we show the average running time of performing
20K random walk steps for different methods when estimating $5$-node graphlet
concentration.  Among them, \textsf{Exact} represents the enumeration time of
the method~\cite{hovcevar2014combinatorial}. Since \textsf{SRW3CSS} incurs
high computation cost, we do not consider this algorithm.  The time cost of the
random walk  consists of populating a random neighbor and identifying the
graphlet types. Here we do not consider the APIs response delay.  From this
table, we know that the random walk on $G^{(d)}$ is faster when $d$ is smaller. For
example, \textsf{SRW2} is much faster than \textsf{SRW3}
and \textsf{SRW4}.  This also validates the choice of smaller $d$.
\begin{table}[t]
\centering
\caption{Running time of performing 20K random walk steps for different methods. } \label{table:timecost}
\resizebox{0.48\textwidth}{!}{
\begin{tabular}{|l|c|c|c|c|c|}
    \hline
    Graph & \textsf{SRW2}  & \textsf{SRW2CSS} & \textsf{SRW3} & \textsf{SRW4} & \textsf{Exact}\\ \hline \hline
    BrightKite~\cite{konect} 
    & 19.4 ms & 110.2 ms & 271.1 ms & 20.6 s & 511 s\\
    Epinion~\cite{snapnets}    
    & 20.6 ms & 68.6 ms & 540.0 ms &51.4 s & 11091 s\\
    Slashdot~\cite{snapnets}   
    & 19.6 ms & 50.6 ms & 538.8 ms & 47.4 s & 5702 s\\
    Facebook~\cite{konect}  
    & 21.8 ms & 114.8 ms & 214.2 ms & 19.8 s & 4405 s\\\hline
    \end{tabular}
}
\end{table}

\subsection{Comparison with Competing Methods}
\subsubsection{Methods with restricted access assumption}
Our framework is mainly designed for graphs with restricted access, i.e., the
graph data can only be accessed through APIs.  For graphs with restricted
access, random walk-based methods are commonly used to exploit the properties
of the graphs. In this part, we compare the best methods in our framework with
two state-of-the-art random walk-based methods~\cite{wang2014efficiently}
and~\cite{hardiman2013estimating}.

The state-of-the-art random walk-based method in estimating any $k$-node
graphlet concentration is the \textsf{PSRW} proposed
in~\cite{wang2014efficiently}. Note that \textsf{PSRW} is equivalent to
choosing $d = k-1$ in our framework.  Specifically, \textsf{PSRW} corresponds
to \textsf{SRW2}, \textsf{SRW3}, and \textsf{SRW4} when estimating $3, 4$ and
$5$-node graphlet concentration respectively. We recap the comparison results
in Figure~\ref{fig:different_random_walks} and Table~\ref{table:timecost} as
follows.
\begin{itemize}[leftmargin=*]
    \item For $3$-node graphlets, \textsf{SRW1CSSNB} performs best for all the
        datasets in Table~\ref{table:dataset}, and it outperforms \textsf{PSRW}
        up to 3.8$\times$ (for ``Twitter''). For $4, 5$-node graphlets,
        \textsf{SRW2CSS} performs better than \textsf{PSRW} both in time cost
        and accuracy, e.g., \textsf{SRW2CSS} outperforms \textsf{PSRW} in
        estimating $4$-node graphlet concentration {\em up to an order of
        magnitude} (for ``Gowalla'') in accuracy.  
        We have consistent better performance in \textsf{SRW1CSSNB} for
        $3$-node graphlets and \textsf{SRW2CSS} for $4, 5$-node graphlets.
\end{itemize}

The triangle concentration has strong relationship with the
global clustering coefficient which is defined as
$3C^3_2 / (C^3_1 + 3C^3_2) = 3c^3_2 / (2c^3_2 + 1)$. Similarly, $c^3_2$ can be
computed directly with the clustering coefficient. Hence we also consider the
method proposed by Hardiman et
al.~\cite{hardiman2013estimating}, which is primarily designed for clustering
coefficient, as the competing method for $3$-node
graphlet concentration. The method in~\cite{hardiman2013estimating} uses the
simple random walk on $G$, and at each step, the visited node $v_t$ checks
whether the node visited before $v_t$ and the node visited after $v_t$ are
connected. {\em Detailed analysis of this method reveals that it is equivalent to
\textsf{SRW1} under our framework}, but derived in a totally different way.
From Figure~\ref{fig:different_random_walks},  we observe that our
method \textsf{SRW1CSSNB}
outperforms~\cite{hardiman2013estimating} (\textsf{SRW1}). Especially for
Wikipedia and Sinaweibo, \textsf{SRW1CSSNB} has at least $3\times$
smaller NRMSE.

\subsubsection{Methods with full access assumption}
Now we assume the graph data is readily available and fits in the main memory.
Our framework is primarily designed for graphs with restricted access.
We evaluate our framework in such full access setting with the purpose to shed
light on the advantages and disadvantages of applying the MCMC samplers for
memory-based graphs.  We compare with two state-of-the-art methods: wedge
sampling~\cite{seshadhri2013triadic}, and path sampling~\cite{jha2015path}.

\noindent{\bf Wedge sampling~\cite{seshadhri2013triadic}.} This method
estimates the triadic measures (e.g., number of
triangles) by generating uniform wedge (\scalebox{0.4}{\GTHREE{1}}) samples.
To get a uniform wedge sample, it first selects a random node $v$ according to
the probability $p_v$, here $p_v\triangleq {d_v\choose 2}/\left(\sum_{u\in V}
{d_u\choose 2}\right)$, and then chooses a uniform random pair of neighbors of
$v$ to generate a wedge. 
Note that wedge sampling needs to compute the probability $p_v$ for each node
$v$. 
This preprocess step has time complexity $O(|V|)$. Besides, the time
complexity of getting $n$ wedge samples is $O(n\log |V|)$. Hence the time
complexity of wedge sampling is $O(|V| + n\log |V|)$.

\noindent{\bf Path sampling~\cite{jha2015path}.} It estimates
$4$-node graphlet counts by generating uniform $3$-path
(\adjustbox{valign=m}{\scalebox{0.4}{\Gfour{1}}}) samples. For each edge $e =
(u, v)\in E$, denote $\tau_e = (d_u - 1)(d_v - 1)$ and $S \triangleq \sum_{e}\tau_e$. To get a uniform $3$-path, it first selects an edge $e=(u, v)$ with
probability $p_e\triangleq \tau_e / S$. Then picks uniform random neighbor
$u^{\prime}$ of $u$ other than $v$, picks uniform random neighbor $v^{\prime}$
of $v$ other than $u$, and outputs the three edges $\{(u^{\prime}, u), (u, v), (v,
v^{\prime})\}$ as the sampled $3$-path. 
The preprocess time of the algorithm is $O(|E|)$. The time of generating $n$
samples is $O(n\log |E|)$. The author also proposed an improved sampler to
estimate graphlets counts.
\begin{figure}[t]
    \centering
    \setcounter{subfigure}{0}
    \subfloat[][Triangle ($g^3_2$, \scalebox{0.4}{\GTHREE{2}})]
    {\label{subfig:3nodeCount}
        \includegraphics[width=0.48\textwidth, height=0.15\textwidth]{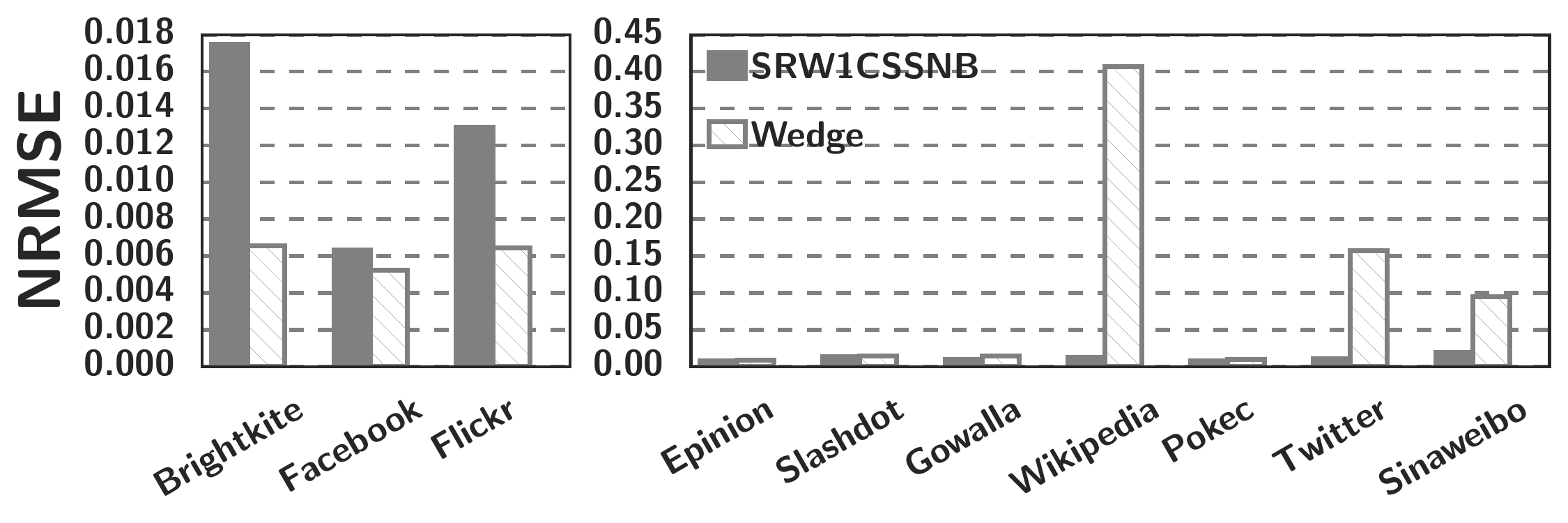}}\\
    \subfloat[][$4$-node clique ($g^4_6$, \scalebox{0.4}{\Gfour{6}})]
    {\label{subfig:4nodeid6Count}
        \includegraphics[width=0.48\textwidth, height=0.15\textwidth]{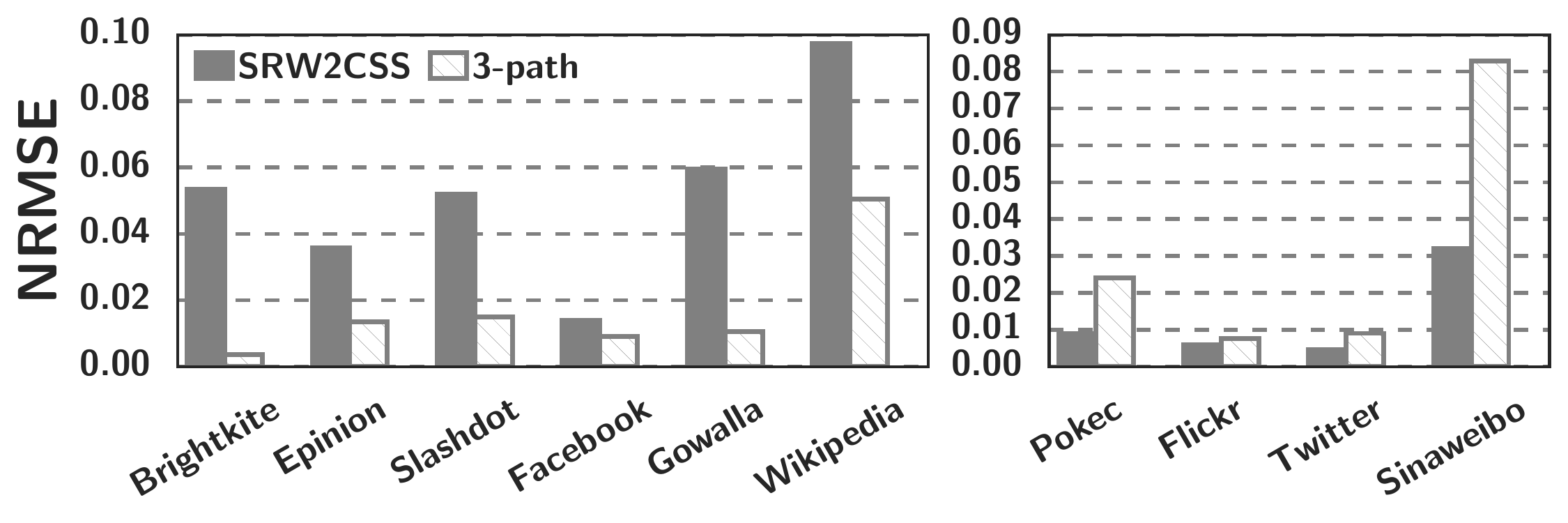}}
    \caption{NRMSE of graphlet counts estimation.}\label{fig:compare_full_access}
\end{figure}

{\em Note that these two methods focus on estimating graphlet counts.}
Our framework is also capable of estimating graphlet counts for memory-based graphs
according to Equation~\eqref{eq:graphletcounter}
and~\eqref{eq:css_unbiased_estimator}. 
We compare \textsf{SRW1CSSNB} with wedge sampling, and \textsf{SRW2CSS} with
path sampling for $3, 4$-node graphlet counts estimation.  For fair comparison,
we run wedge sampling and path sampling for 200K samples (200K samples are
used in the original papers), and then run the
\textsf{SRW1CSSNB} and \textsf{SRW2CSS} for the same running time as wedge
sampling and path sampling respectively. The results are shown in
Figure~\ref{fig:compare_full_access}. For path sampling, we only show the
accuracy of $4$-node clique due to space limitation. We summarize the results
in Figure~\ref{fig:compare_full_access} as follows.
\begin{itemize}[leftmargin=*]
    \item In Figure~\ref{subfig:3nodeCount}, we compare our method
        \textsf{SRW1CSSNB} with wedge sampling (denoted as \textsf{Wedge} in
        the figure). For graphs BrightKite, Facebook, and Flickr,
        \textsf{Wedge} is more accurate than \textsf{SRW1CSSNB} given the same
        running time. Actually, these three graphs have the highest triangle
        concentration in the datasets.  \textsf{SRW1CSSNB} performs better than
        \textsf{Wedge} given the same running time for the rest graphs.  Both
        of \textsf{SRW1CSSNB} and \textsf{Wedge} generate wedge samples. The
        difference is that \textsf{Wedge} generates independent wedge samples
        while \textsf{SRW1CSSNB} generates correlated samples.  However,
        \textsf{SRW1CSSNB} generates samples much faster than \textsf{Wedge}.
        Hence for large graphs with small triangle concentration, we recommend
        to use the \textsf{SRW1CSSNB} to estimate the triangle counts.
    \item In Figure~\ref{subfig:4nodeid6Count}, we evaluate our proposed method
        \textsf{SRW2CSS} and path sampling (which is denoted as \textsf{3-path}
        in the figure). For the large graphs Pokec, Flickr, Twitter, and
        Sinaweibo, our method \textsf{SRW2CSS} performs better than
        \textsf{3-path}. This is because \textsf{SRW2CSS} does not need a
        preprocess step and generates samples much faster than
        \textsf{3-path}. For large graphs with more than 100M edges, we
        recommend to use our method \textsf{SRW2CSS}. 
\end{itemize}
\noindent \textbf{Summary:} Our framework is designed for graphs with
restricted access. However, the simulation results in
Figure~\ref{fig:compare_full_access} indicate that our proposed methods are
also more efficient than the state-of-the-art methods for large memory-based
graphs.  Besides, both of \textsf{Wedge} and \textsf{3-path} are not general
and not directly applicable for graphs with restricted access. 

\subsubsection{Adapted memory-based methods}
In this part, we adapt the wedge sampling~\cite{seshadhri2013triadic} to the
random walk-based method and compare the adapted method with our framework.  We
do not adapt the path sampling because the $3$-path sampler in path sampling
cannot output the samples of $3$-star (\scalebox{0.3}{\Gfour{2}}) directly, and
hence the adaption of the path sampling is much more complicated. 

\noindent\textbf{Adapted wedge sampling.} Wedge sampling cannot be applied for
graphs with restricted access directly. We give an adaption of the wedge
sampling so as to make it applicable for the restricted setting. The main idea
of the adaption is to sample each node $v$ in the graph with probability
$p_v\triangleq {d_v\choose 2}/\left(\sum_{u\in V} {d_u\choose 2}\right)$ via
the Metropolis-Hasting random walks (MHRW)~\cite{metropolis1953equation,
hastings1970monte}, and then choose a uniform random pair of neighbors of the
sampled node $v$ to generate a wedge.
Algorithm~\ref{algo:adapted_wedge_sampling} in Appendix~\ref{app:adaption}
presents the pseudo code of the adaption for wedge sampling. 
According to Algorithm~\ref{algo:adapted_wedge_sampling}, the adapted wedge
sampling method needs to calling the APIs of the underlying networks for three
times at each random walk step, which means \textit{the adapted method has
$3\times$ API's calling cost compared with our method when given the same
random walk steps}.

The comparison between \textsf{SRW1CSSNB} and the adapted wedge sampling
\textsf{Wedge-MHRW} is demonstrated in Figure~\ref{fig:compare_adapted_wedge}. 
We compare their accuracy in estimating the triangle concentration.
The NRMSE is estimated over 1,000
independent simulations. The results are summarized as follows. 
\begin{itemize}[leftmargin=*]
    \item Figure~\ref{subfig:3nodeConcentrationCompare} shows that our proposed
        method \textsf{SRW1CSSNB} has much higher accuracy than the
        \textsf{Wedge-MHRW} when estimating the triangle concentration with
        the same number of random walk steps.  For
        example, the NRMSE of \textsf{SRW1CSSNB} is $8\times$ smaller than
        \textsf{Wedge-MHRW} for graph Wikipedia when estimating the triangle
        concentration using $20$K random steps. 
    \item In Figure~\ref{subfig:3nodeConcentrationConvergenceCompare}, we
        choose two largest graphs in our datesets to compare the convergence of
        \textsf{SRW1CSSNB} and \textsf{Wedge-MHRW}. Both of these two methods
        converge to the ground truth value when increasing the number of random
        walk steps. Besides, our method has consistent higher accuracy
        than \textsf{Wedge-MHRW}.
\end{itemize}
\begin{figure}[t]
    \centering
    \setcounter{subfigure}{0}
    \subfloat[][Accuracy (random walk steps = 20K)]
    {\label{subfig:3nodeConcentrationCompare}
        \includegraphics[width=0.48\textwidth, height=0.15\textwidth]{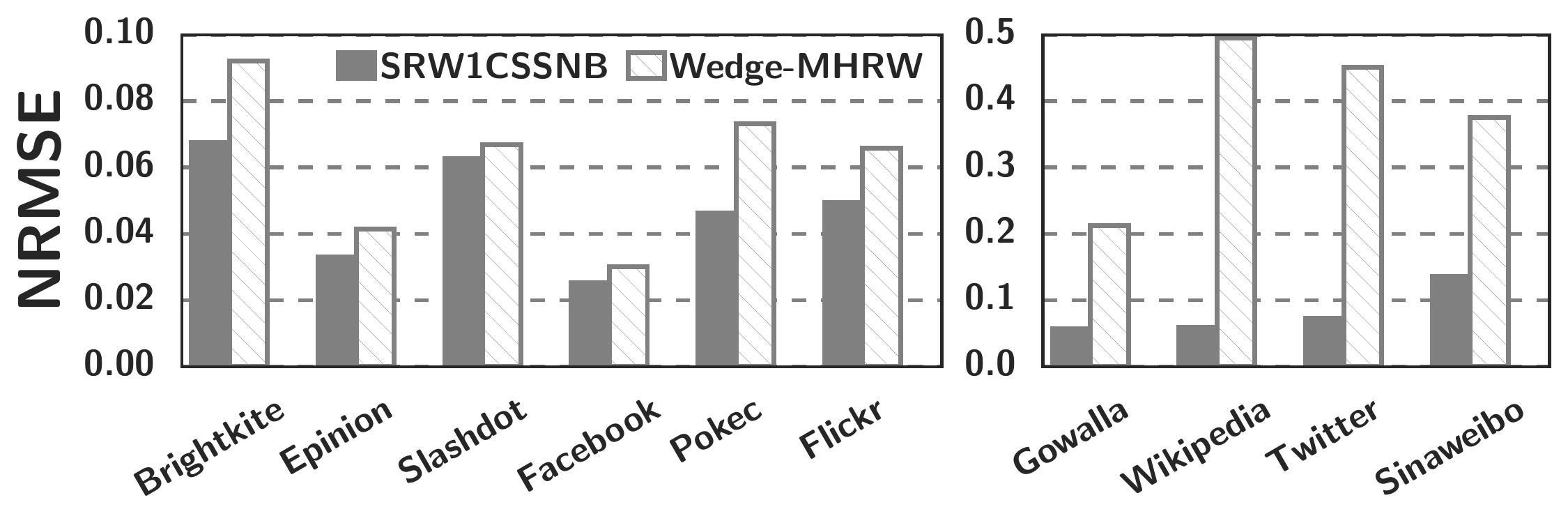}}\\
    \subfloat[][Convergence]
    {\label{subfig:3nodeConcentrationConvergenceCompare}
        \includegraphics[width=0.23\textwidth,
        height=0.14\textwidth]{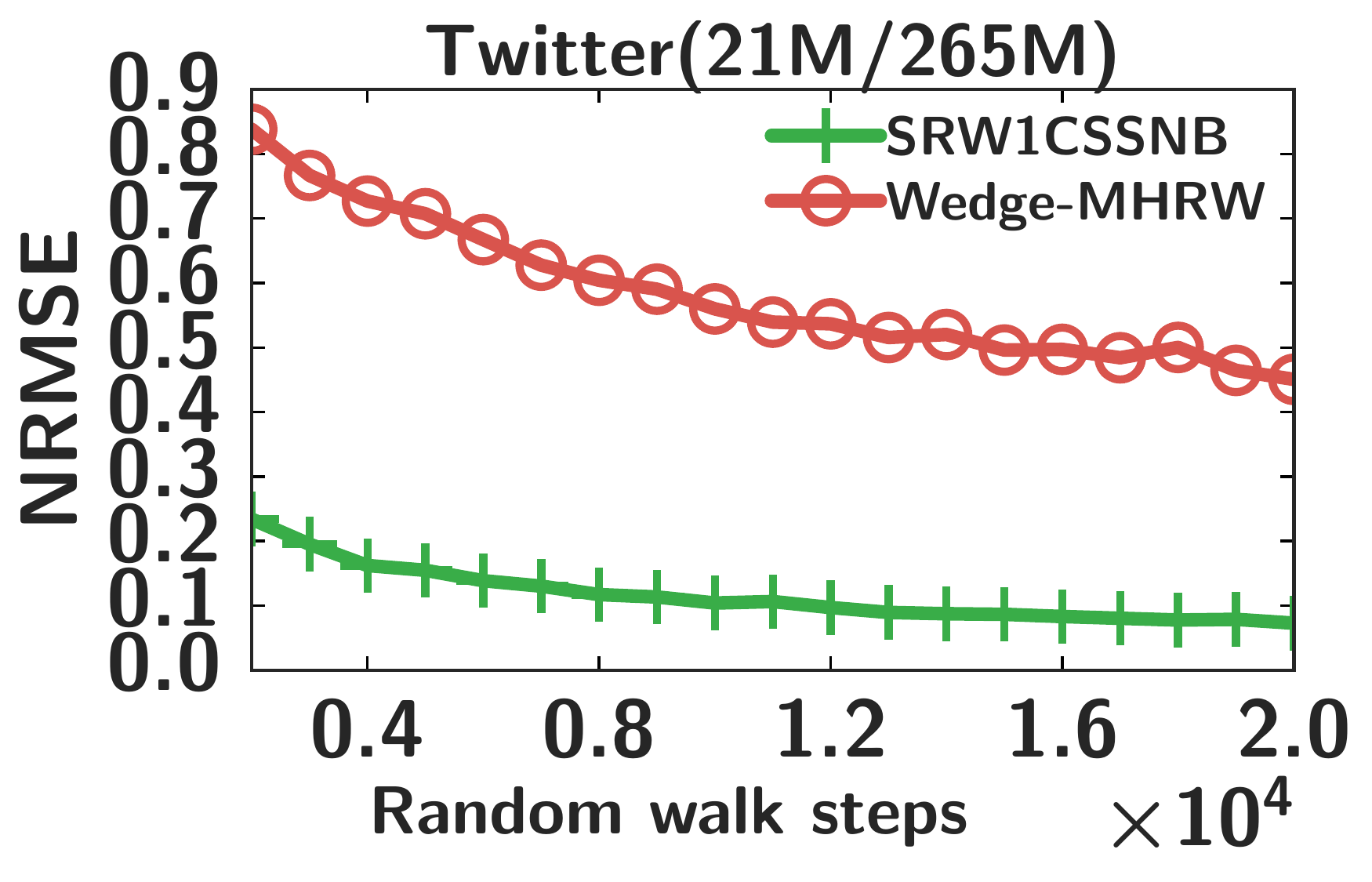}
        \includegraphics[width=0.23\textwidth, height=0.15\textwidth]{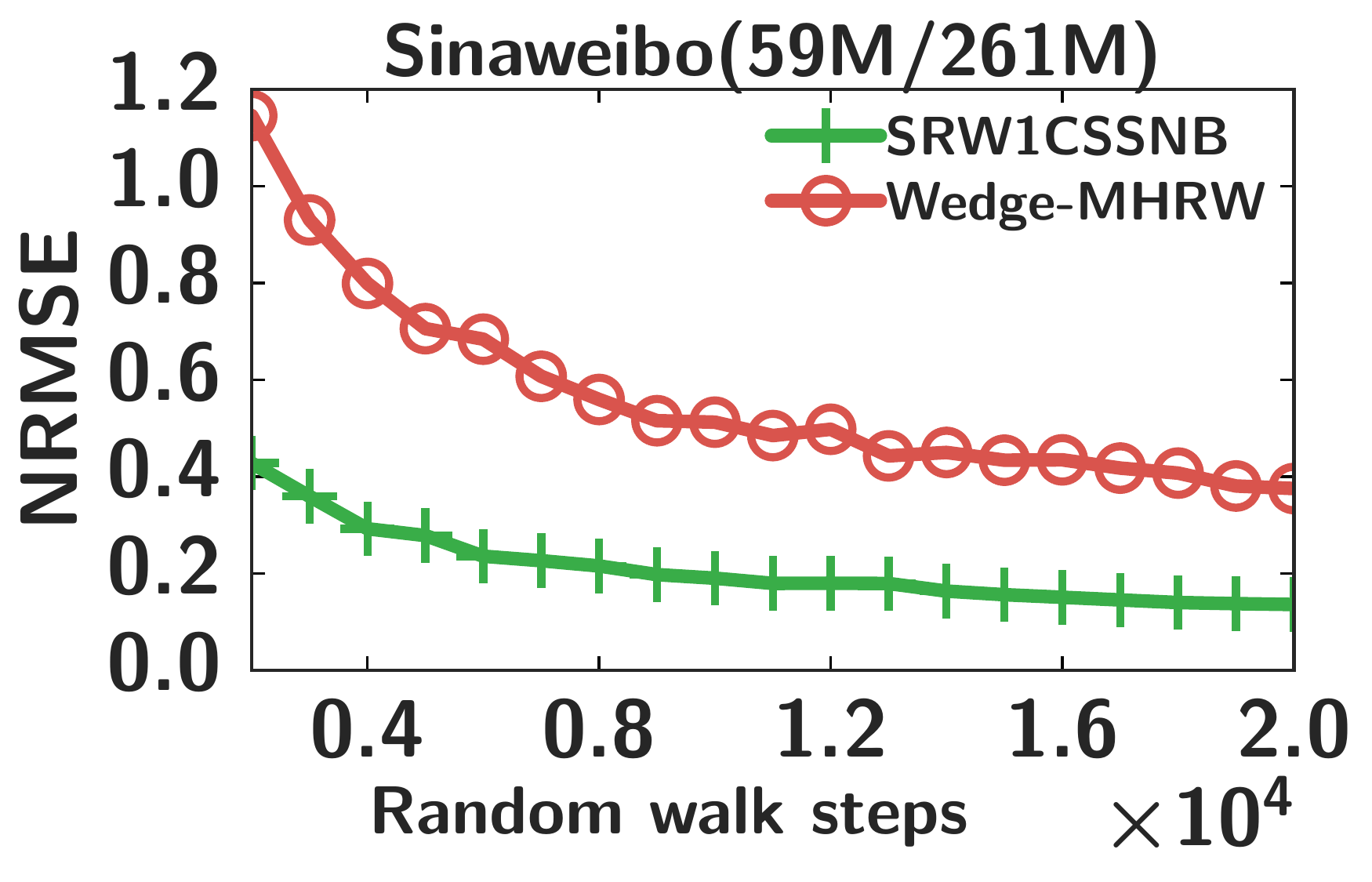}}
    \caption{Comparison between our proposed method \textsf{SRW1CSSNB} and the
    adapted wedge sampling \textsf{Wedge-MHRW} in estimating the triangle
    concentration.}\label{fig:compare_adapted_wedge}
\end{figure}

\subsection{Applications}
In this subsection, we apply our framework to analyze the intrinsic properties
of the large scale graph Sinaweibo in our datasets. Sinaweibo is the most
popular microblog service in China and has reached 222 million monthly active
users as of September 2015. It allows users to follow others, post and repost
messages, comment on others' posts, etc.  With the fact that Facebook is an
online social network while Twitter is more like a news
media~\cite{kwak2010twitter}, we now study whether Sinaweibo acts like a social
network or a news media by measuring its similarity to Twitter and Facebook. To
measure the similarity, we adopt the definition of graphlet kernel
in~\cite{5664} and restrict the definition to $4$-node graphlets. Specifically,
for two graphs with $4$-node graphlet concentration $\mathbf{c}_1$ and
$\mathbf{c}_2$, we define the similarity between them as $\mathbf{c}^{T}_1\cdot
\mathbf{c}_2 / (\norm{\mathbf{c}_1}\cdot \norm{\mathbf{c}_2})$. We run the
random walk for 20K steps for graphs Facebook, Twitter, and Sinaweibo to
estimate the $4$-node graphlet concentration and compute the similarity with
the estimated value. The results of 100 simulations are reported in
Table~\ref{table:similarity}. Our proposed method \textsf{SRW2CSS} gives a more
accurate estimate of similarity value compared with \textsf{PSRW}. Besides, we
find that Sinaweibo has more similar subgraph building blocks to Twitter, which
indicates Sinaweibo also acts like an efficient platform for information
diffusion. Note that for these large graphs, it is impractical to crawl the
whole datasets for later on analysis. Our framework produces accurate estimates
based on a small portion of crawled nodes, which implies our framework is an efficient analysis
tool for large graphs with restrict access.
\begin{table}[t]
\centering
\caption{Similarity between Sinaweibo and social network Facebook,
    as well as news media Twitter. } \label{table:similarity}
\resizebox{0.45\textwidth}{!}{
\begin{tabular}{|l|c|c|c|}
    \hline
    Graph & \textsf{SRW2CSS} & \textsf{PSRW} & \textsf{Exact}\\ \hline \hline
    Facebook~\cite{konect} 
    & {\bf 0.5809$\pm$0.0501} & 0.5856$\pm$0.0676 & 0.5757 \\
    Twitter~\cite{nr}    
    & {\bf 0.9988$\pm$0.0236} & 0.9957$\pm$0.0200 &0.9999 \\
    \hline
    \end{tabular}
}
\end{table}

\section{Conclusion}\label{sec:conclusion}
In this paper, we propose a novel random walk-based framework which takes one
tunable parameter to estimate the graphlet concentration. Our framework is
general and can be applied to any $k$-node graphlets. We derive an analytic
bound on the number of random walk steps required for convergence. We also
introduce two optimization techniques to further improve the efficiency of our
framework.  Our experiments with many real-world networks show that the methods
with appropriate parameter in our framework produce accurate estimates and
outperform the state-of-the-art methods significantly in estimation accuracy.

\balance

\section{Acknowledgments}
We thank the reviewers for their valuable comments. 
The work of Yongkun Li was sponsored  by CCF-Tencent Open Research Fund.
Pinghui Wang was supported by Ministry of Education \& China Mobile Joint
Research Fund Program (MCM20150506), the National Natural Science Foundation of
China (61603290, 61103240, 61103241, 61221063, 91118005, 61221063, U1301254),
Shenzhen Basic Research Grant (JCYJ20160229195940462), 863 High Tech
Development Plan (2012AA011003), 111 International Collaboration Program of
China,  and the Application Foundation Research Program of SuZhou (SYG201311).
The work of John C.S. Lui is supported in part by RGC 415013 and Huawei Research Grant.

\bibliographystyle{abbrv}
\bibliography{mybib} 

\normalsize 

\begin{appendix}
\section{Proof of
Theorem~\ref{thm:stationary_distribution}}\label{app:section:stationary_distribution}
\stationarydistribution*
\begin{proof}
    We first prove that there exists unique stationary distribution for the
    expanded Markov chain. For any two states $X^{(l)}_i = (X_{i_1}, \cdots,
    X_{i_l})$ and $X^{(l)}_j = (X_{j_1}, \cdots, X_{j_l})$ in
    $\mathcal{M}^{(l)}$, $X_{i_l}$ can reach state $X_{j_ 1}$ in finite steps
    during the random walk on $G^{(d)}$. This indicates that $X^{(l)}_i$ can
    reach $X^{(l)}_j$ in finite steps in the expanded Markov chain. Similarly,
    $X^{(l)}_j$ can reach $X^{(l)}_i$ in finite steps. We conclude that the
    expanded Markov chain is irreducible.  Since any irreducible Markov chain
    has one unique stationary distribution~\cite[Theorem 5.3]{olle2000finite}, there exists
    a unique stationary distribution for the expanded Markov chain. 
    
    To derive the closed form of $\PIMB$, we write the entry of
    transition matrix $\mathbf{P}_e$ of the expanded Markov chain for state
    $X^{(l)}_i$ and $X^{(l)}_j$ when $l > 1$ as follows: 
    {\small
    \begin{equation*}
        \mathrm{P}_e(X^{(l)}_i, X^{(l)}_j) = \begin{cases}
            \frac{1}{d_{X_{i_l}}} & \begin{split}
            &\text{if }(X_{i_l}, X_{j_{l}})\in R^{(d)},\\
            &\text{and } X_{i_q} = X_{j_{q-1}},1<q\leq l,
        \end{split}
            \\
            0 &\text{otherwise}.
        \end{cases}
    \end{equation*}
    }
    For $\PIMB$ in Equation~(\ref{eq:stationary_distribution}) we can
    verify that 
    \begin{equation*}
         \PIMB =\PIMB\cdot \mathbf{P}_e \quad \mbox{and} \quad
            \sum_{X^{(l)}\in \mathcal{M}^{(l)}} \pi_e(X^{(l)}) = 1.
    \end{equation*}
    According to Theorem 5.1 and 5.3 in~\cite[Chapter 5]{olle2000finite},
    $\PIMB$ in
    Equation~(\ref{eq:stationary_distribution}) is the unique stationary
    distribution.
\end{proof}

\section{Computation of Coefficient}\label{app:computation_of_coefficient}

In this section, we discuss in detail how to compute $\alpha^k_i$.  The
coefficient $\alpha^k_i$ is the key to our framework.  One can
view it as part of the ``{\em inclusion probability}''.  Here we provide a
method to compute $\alpha^k_i$ in Algorithm~\ref{algo:compute_coefficient}
for $SRW(d)$ when  $1\leq d \leq k - 1$. In Line~\ref{algo:line:condition}
of Algorithm~\ref{algo:compute_coefficient}, $S^{(l)}_{x_i}$ represents the
$x_i$-th element in list $S^{(l)}$. A special case of
Algorithm~\ref{algo:compute_coefficient} is $d = k - 1$. In this case,
$\alpha^k_i$ can also be computed by enumerating $(k - 1)$-node induced
subgraphs of $g^k_i$, which has time complexity $O(k^3)$.  The coefficient
$\alpha^k_i$ equals to $(|S| - 1)\cdot |S|$, here $S$ is the set of $(k -
1)$-node subgraphs in $k$-node graphlet $g^k_i$.  
\begin{algorithm}[H]
\caption{Pseudo code of computing $\alpha^k_i$}
\begin{algorithmic}[1]\label{algo:compute_coefficient}
    \REQUIRE node set and edge list of graphlet $g^k_i$, $SRW(d)$
    \ENSURE$\alpha^k_i$ for graphlet $g^k_i$
    \STATE{Set $S\leftarrow$ all $d$-node connected induced subgraphs of $g^k_i$}\\
    \STATE{Random walk block length $l\leftarrow k - d + 1$ }\\
    \STATE{Counter $\alpha^k_i\leftarrow 0$}\\
    \FORALL{$S^{(l)}\leftarrow$ combination of $l$ elements from $S$}
    \IF{size of node set $\cup_{s\in S^{(l)}}V(s)$ less than $k$}
    \STATE{\textbf{continue}}
        \ENDIF
        \FORALL{$(x_1, \cdots, x_l)\leftarrow$ permutation of $1, \cdots, l$}
        \IF {$S^{(l)}_{x_i}$ and $S^{(l)}_{x_{i+1}}$ share $d\!-\!1$ nodes,
        $\forall1\!\leq\!i\!\leq\!l\!-\!1$}\label{algo:line:condition}
        \STATE{$\alpha^k_i\leftarrow \alpha^k_i +
        1$}\label{algo:line:incremental}
        \ENDIF
    \ENDFOR
    \ENDFOR
    \RETURN $\alpha^k_i$
\end{algorithmic}
\end{algorithm}

\section{Chernoff Bound for Estimator}\label{app:chernoff_bound}
Our proof is based on the following Chernoff-Hoeffding bound for finite state
Markov chain.
\begin{theorem}\label{thm:chernoff_bound}
\cite[Theorem 3]{chung2012chernoff} For a finite and ergodic Markov chain with
    state space $\mathcal{N}$ and stationary distribution $\pmb{\pi}$, let
    $\tau=\tau(\varepsilon)$ denote its $\varepsilon$-mixing time  for $\varepsilon
    \leq 1/8$.  Let $X_1, \cdots, X_n$ denote a $n$-step random walk starting
    from an initial distribution $\varphi$ on $\mathcal{N}$. Define
    $\norm{\varphi}_{\pmb{\pi}}=\sum_{i=1}^{|\mathcal{N}|}\frac{\varphi_i^2}{\pi_i}$.
    The expectation of function $f:\mathcal{N}\rightarrow [0, 1]$ is denoted by
    $\mathbb{E}_{\pmb{\pi}}[f(X)]=\mu$. Define the total weight of the walk
    $X_1, \cdots, X_n$ by $Z \triangleq \sum_{s=1}^{n}f(X_s)$. There exists
    some constant $c$ which is independent of $\mu, \varepsilon$ and $\epsilon$ such
    that for $0<\delta<1$
\begin{equation*}
\Pr\left[\left|\frac{Z}{n}-\mu\right| >\epsilon\mu\right] \leq
c\norm{\varphi}_{\pmb{\pi}}e^{-\epsilon^2\mu n / 72\tau}
\end{equation*}
\end{theorem}

%

In the following, we first prove that the numerator $\hat{C}^k_i\triangleq
\frac{1}{n}\sum_{s=1}^{n}h^k_i(X^{(l)}_s)/\left(\alpha^k_i\pi_e(X^{(l)}_s)\right)$
in Equation~\eqref{eq:concentration_estimator} concentrates around its expected
value $C^k_i$.
\begin{lemma}\label{lemma:bound_numerator}
    For $0 < \delta < 1$, there exists a constant $\xi$, such that if $n \geq
    \xi(\frac{W}{\alpha^k_i
    C^k_i})\frac{\tau^\prime}{\epsilon^2}(\log
    \frac{\norm{\varphi}_{\PIMB}}{\delta})$ we have
    \begin{equation*}
    \Pr\left[\left|\hat{C}^k_i - C^k_i\right|/C^k_i>\frac{\epsilon}{3}\right] <
    \delta / 2
    \end{equation*}
    where $W$ is defined as
    $\max_{X^{(l)}\in\mathcal{M}^{(l)}}1/\pi_e(X^{(l)})$,
    $\varphi$ is the initial distribution, $\norm{\varphi}_{\PIMB}$ is defined
    as $\sum_{X^{(l)}}\varphi^2(X^{(l)})/\pi_e(X^{(l)})$, and $\tau^\prime$ is the
    mixing time $\tau^{\prime}(1/8)$ of the expanded Markov chain.
\end{lemma}
\begin{proof}
    Let $f_i=\frac{h^k_i(X^{(l)})/\left(\pi_e(X^{(l)})\right)}{W}$ such that all values
    of $f_i$ are in $[0, 1]$.  The expectation of $f_i$ is $\mu_i =
    \mathbb{E}_{\PIMB}[f_i] = \frac{\alpha^k_i C^k_i}{W}$.  Using the
    result in Theorem~\ref{thm:chernoff_bound}, we have
    \begin{equation*}
    \Pr\left[\left|\hat{C}^k_i/C^k_i - 1\right|>\frac{\epsilon}{3}\right]
    < c\norm{\varphi}_{\PIMB}e^{-\epsilon^2\mu_i n / 9\cdot72\tau^\prime}
    \end{equation*} 
    Extracting $n$ for which $\frac{\delta}{2} =
    c\norm{\varphi}_{\PIMB}e^{-\epsilon^2\mu_i n / 9\cdot72\tau^\prime}$, we
    have $n \geq
    \xi(\frac{W}{\alpha^k_i C^k_i})\frac{\tau^\prime}{\epsilon^2}(\log
    \frac{\norm{\varphi}_{\PIMB}}{\delta})$
\end{proof}

We then prove that 
\begin{equation*}
\hat{C}^k\triangleq
\frac{1}{n}\sum_{s=1}^{n}h^k(X^{(l)})/\left(\alpha^k(X^{(l)}_s)\pi_e(X^{(l)}_s)\right)
\end{equation*}
concentrates around its expected value $C^k\triangleq
\sum_{i=1}^{|\mathcal{G}^k|}C^k_i$.

\begin{lemma}\label{lemma:bound_denominator}
    For $0 < \delta < 1$, there exists a constant $\xi$, such that if $n \geq
    \xi(\frac{W}{\alpha_\text{min}C^k})\frac{\tau^\prime}{\epsilon^2}(\log
    \frac{\norm{\varphi}_{\PIMB}}{\delta})$ we have
    \begin{equation*}
    \Pr\left[\left|\hat{C}^k - C^k\right|/C^k>\frac{\epsilon}{3}\right] <
    \delta / 2
    \end{equation*}
    where $W$ is defined as
    $\max_{X^{(l)}\in\mathcal{M}^{(l)}}1/\pi_e(X^{(l)})$ and
    $\alpha_\text{min}=\min_j \alpha^k_j$. $\tau^\prime$ is the
    mixing time $\tau^\prime(1/8)$ of the expanded Markov chain starting from the initial
    distribution $\varphi$.
\end{lemma}
\begin{proof}
    We define $f =
    \frac{h^k(X^{(l)})/\left(\alpha^k(X^{(l)})\pi_e(X^{(l)})\right)}{W/\alpha_\text{min}}$
    if $\alpha_\text{min}\neq 0$. Otherwise we define $f = 0$. A single line
    calculation shows that $\mu = \mathbb{E}_{\PIMB}[f] =
    \frac{\alpha_\text{min}C^k}{W}$. Following the same argument in
    Lemma~\ref{lemma:bound_numerator} we can derive the bound of $n$.
\end{proof}

\unionbound*
\begin{proof}
    We first find steps $n_i$ to guarantee $\hat{C}^k_i$ and $\hat{C}^k$ be
    within $\epsilon/3$ deviation from their expectation with probability
    greater than $1 - \delta / 2$. See Lemma~\ref{lemma:bound_numerator}
    and~\ref{lemma:bound_denominator} for more details. Then use the fact
    \begin{equation*}
        (1-\epsilon) c^k_i\leq \frac{(1-\epsilon / 3)C^k_i}{(1+\epsilon/3)C^k}
        \leq \frac{\hat{C}^k_i}{\hat{C}^k} \leq \frac{(1 + \epsilon /
        3)C^k_i}{(1 - \epsilon / 3)C^k} \leq (1 + \epsilon) c^k_i
    \end{equation*}
    we can prove that with  $n \geq
    \xi(\frac{W}{\Lambda})\frac{\tau^\prime}{\epsilon^2}(\log
    \frac{\norm{\varphi}_{\PIMB}}{\delta})$, the estimator $\hat{c}^k_i$
    is within $\epsilon$ deviation from expectation with probability greater
    than $1-\delta$. Here $\tau^\prime$ is the mixing time $\tau^{\prime}(1/8)$
    of the expanded Markov chain. Second, we prove that $\tau^\prime = \tau$.
    To see this, we denote the initial distribution of the random walk as
    $\pmb{\pi}_0$. After $l -1$ steps, we get the initial distribution for the
    expanded Markov chain. At time $t$, we denote the distribution of random
    walk as $\pmb{\pi}_t$ and distribution for expanded Markov chain is
    $\overline{\pmb{\pi}}_{t - l + 1}$.  At time $t + l - 1$, we get
    distribution $\overline{\pmb{\pi}}_t$, for any state $X^{(l)} = (X_1,
    \cdots, X_{l})$. Since $\overline{\pi}(X^{(l)}) =
    \pi_t(X_1)\frac{1}{d_{X_1}}\cdot\frac{1}{d_{X_{l-1}}}$, the variation
    distance satisfies $|\pmb{\pi} - \pmb{\pi}_t|_1 = |\pmb{\pi}_e -
    \overline{\pmb{\pi}}_t|_1$. Thus we have $\tau = \tau^\prime$. This ends
    the proof.
\end{proof}

\section{Proof of Lemma~\ref{lemma:css_variance}}\label{app:css_variance}
\cssvariance*
\begin{proof}
    Variance of random variable $X$ can be expanded as
    $\mathbb{E}[X^2]-(\mathbb{E}[X])^2$. Since the expectations of both
    functions are equal, we only need to prove that
    \begin{equation*}
        \mathbb{E}_{\PIMB}\left[\left(h^k_i(\X{l})/p(X^{(l)})\right)^2\right]\leq
        \mathbb{E}_{\PIMB}\left[\left(h^k_i(\X{l})/(\alpha^k_i\pi_e(X^{(l)}))\right)^2\right]
    \end{equation*}
    Note that the left hand side of the inequality can be rewritten as
    \begin{equation*}
        \frac{1}{(2|R^{(d)}|)^2}\sum_{s\simeq g^k_i} \frac{1}{\sum_{\X{l}_j\in
        \mathcal{C}(s)}\pi_e(\X{l}_j)}
    \end{equation*}
    while the right hand side can be represented as
    \begin{equation*}
        \frac{1}{(2|R^{(d)}|)^2}\sum_{s\simeq g^k_i} \frac{1}{(\alpha^k_i)^2}
        \sum_{\X{l}_j\in \mathcal{C}(s)}1 / \pi_e(\X{l}_j).
    \end{equation*}
    Use the fact that {\em harmonic mean is not greater than arithmetic mean}, we
    have
    \begin{equation*}
        \frac{1}{\sum_{\X{l}_j\in
        \mathcal{C}(s)}\pi_e(\X{l}_j)}\leq
        \frac{1}{(\alpha^k_i)^2}
        \sum_{\X{l}_j\in \mathcal{C}(s)}1 / \pi_e(\X{l}_j).
    \end{equation*}
    This ends the proof.
\end{proof}

\section{Computation of sampling probability}\label{app:computation_of_sampling_probability}
Algorithm~\ref{algo:css_reweight} shows the pseudo code to compute the
$p(\X{l})$ ($l > 2$) when given node set and edge list of subgraph $s$
induced by nodes in $\X{l}$. 
\begin{algorithm}[H]
    \caption{Pseudo code of computing $p(\X{l})$}
\begin{algorithmic}[1]\label{algo:css_reweight}
    \REQUIRE state $\X{l}$, graphlet size $k$, $SRW(d)$
    \ENSURE sampling probability $p(\X{l})$
    \STATE{random walk block length $l\leftarrow k - d + 1$ }\\
    \STATE{$s\leftarrow$ subgraph induced by nodes in $X^{(l)}$}\\
    \STATE{set $S\leftarrow$ all $d$-node connected induced subgraphs of $s$}\\
    \STATE{counter $p(X^{(l)})\leftarrow 0$}\\
    \FORALL{$S^{(l)}\leftarrow$ combination of $l$ elements from $S$}
    \IF{size of node set $\cup_{s\in S^{(l)}}V(s)$ equals to $k$}
    \FORALL{$(x_1, \cdots, x_l)\leftarrow$ permutation of $1, \cdots, l$}
        \IF {$S^{(l)}_{x_i}$ and $S^{(l)}_{x_{i+1}}$ share $d - 1$ nodes,
        $\forall1\leq\!i\leq\!l\!-\!1$ }
        \STATE{corresponding state ${X^\prime}^{(l)}\leftarrow
        (S^{(l)}_{x_1},\cdots, S^{(l)}_{x_l})$}
        \STATE{$p(X^{(l)})\leftarrow p(X^{(l)}) + \pi_e({X^\prime}^{(l)})$}
        \ENDIF
    \ENDFOR
    \ENDIF
    \ENDFOR
    \RETURN $p(X^{(l)})$
\end{algorithmic}
\end{algorithm}

\section{Adaption of Wedge Sampling}\label{app:adaption}
We design a Metropolis-Hasting random walk (MHRW) to visit each node $v$ with
the stationary distribution $\pi(v) = p_v\triangleq {d_v\choose
2}/\left(\sum_{u\in V} {d_u\choose 2}\right)$. For each sampled node, we choose
a uniform random pair of neighbors to generate (closed or open) wedges.  We
demonstrate the pseudo code of the adapted wedge sampling in
Algorithm~\ref{algo:adapted_wedge_sampling}.  Note that at each random walk
step, the adapted wedge sampling algorithm needs to explore \textit{three
nodes} to obtain their neighborhood information.

\begin{algorithm}[H]
    \caption{Adapted Wedge Sampling}
\begin{algorithmic}[1]\label{algo:adapted_wedge_sampling}
    \REQUIRE sampling budget $n$, Metropolis-Hasting random walk on $G$
    \ENSURE estimate of $c^3_1, c^3_2$
    \STATE{counter $\hat{C}^3_i\leftarrow 0$ for $i \in \{1, 2\}$}\\
    \STATE{random walk steps $t\leftarrow 1$}
    \STATE{$v_1\leftarrow$ randomly chosen node with $d_{v_1} \geq 2$}
    \WHILE{$t \leq n$}
    \STATE{$(v^{\prime}_{t}, v^{\prime\prime}_t)\leftarrow$ uniform random pair of neighbors of $v_t$}
    \IF{$\{v^{\prime}_{t}, v^{\prime\prime}_t, v\}$ induces a triangle}
    \STATE{$\hat{C}^3_2\leftarrow \hat{C}^3_2 + 1$ } \COMMENT{closed wedge}
    \ELSE
    \STATE{$\hat{C}^3_1 \leftarrow \hat{C}^3_1 + 1$} \COMMENT{open wedge}
    \ENDIF
    \STATE{$w\leftarrow$ random neighbor of node $v_t$}
    \STATE{Generate $p\sim U(0, 1)$} \COMMENT{random variable uniformly
    distributed between $0\sim 1$}
    \IF{$p\leq \min\{1, \frac{d_w - 1}{d_{v_t} - 1}\}$}
    \STATE{$v_{t+1}\leftarrow w$}
    \ELSE
    \STATE{$v_{t+1} \leftarrow v_{t}$}
    \ENDIF
    \STATE{$t\leftarrow t + 1$}
    \ENDWHILE
    \RETURN $\hat{c}^3_1 = 3\hat{C}^3_1/(3\hat{C}^3_1 + \hat{C}^3_2),
    \hat{c}^3_2 = \hat{C}^3_2/(3\hat{C}^3_1 + \hat{C}^3_2)$
\end{algorithmic}
\end{algorithm}

\end{appendix}

\end{document}